\documentclass[anz,print]{cupaus}

\usepackage{amsmath}
\usepackage{amsfonts}
\usepackage{amssymb}

\usepackage{graphicx}
\usepackage{comment}
\usepackage{subfig}
\usepackage{bm}

\usepackage{color} 

\def\bmtheta{{\bm{\theta}}}
\def\bmSigma{{\bm{\Sigma}}}
\def\bmmu{{\bm{\mu}}}

\newcommand{\R}{\mathbb{R}}

\newcommand{\transp}{{T}}
\renewcommand{\d}{\hbox{d}}

\newtheorem{df}{Definition}[section]
\newtheorem{theorem}[df]{Theorem}
\newtheorem{corollary}[df]{Corollary}
\newtheorem{proposition}[df]{Proposition}
\newtheorem{remark}[df]{Remark}
\newtheorem{example}{Example}[section]

\begin{document}

\verso{HJB Equation for Constrained Optimal Allocation Problem}
\recto{S. Kilianov\'a and D. \v{S}ev\v{c}ovi\v{c}}

\title{Transformation Method for Solving Hamilton-Jacobi-Bellman Equation for Constrained Dynamic Stochastic Optimal Allocation Problem}

\author{S.~Kilianov\'a}

\cauthormark
\author{D. \v{S}ev\v{c}ovi\v{c}}

\address{
Dept.\ Applied Mathematics \& Statistics, Faculty of Mathematics, Physics and Informatics, Comenius University, 842 48  Bratislava, Slovakia \email[1]{kilianova@fmph.uniba.sk, sevcovic@fmph.uniba.sk}}

\pages{1}{22}

\begin{abstract}
In this paper we propose and analyze a method based on the Riccati transformation for solving the evolutionary Hamilton-Jacobi-Bellman equation arising from the  stochastic dynamic optimal allocation problem. We show how the fully nonlinear Hamilton-Jacobi-Bellman equation can be transformed into a quasi-linear parabolic equation whose diffusion function is obtained as the value function of certain parametric convex optimization problem. Although the diffusion function need not be sufficiently smooth, we are able to prove existence, uniqueness and derive useful bounds of classical H\"older smooth solutions. We furthermore construct a fully implicit iterative numerical scheme based on finite volume approximation of the governing equation. A numerical solution is compared to a semi-explicit traveling wave solution by means of the convergence ratio of the method. We compute optimal strategies for a portfolio investment problem motivated by the German DAX 30 Index as an example of application of the method.
\end{abstract}

\keywords[2000 \textit{Mathematics subject classification}]{Primary: 35K55, Secondary:  
 34E05 70H20 91B70 90C15 91B16}

\keywords[\textit{Keywords and phrases}]{
Hamilton--Jacobi--Bellman equation, Riccati transformation, quasi-linear parabolic equation, finite volume approximation scheme, traveling wave solution}

\maketitle

\section{Introduction}

The purpose of this paper is to propose and analyze a method based on the Riccati transformation for solving a time dependent Hamilton-Jacobi-Bellman equation arising from a stochastic dynamic optimal allocation problem on a finite time horizon, in which our aim is to maximize the expected value of the terminal utility subject to constraints on portfolio composition. 

Investment problems with state constraints were considered and analyzed by Zariphopoulou \cite{Zariphopoulou}, where the purpose was to maximize the total expected discounted utility of consumption for the optimal portfolio investment consisting of a risky and a risk-free asset, over an infinite and finite time horizon. It was shown that the value function of the underlying stochastic control problem is the unique smooth solution to the corresponding HJB equation and the optimal consumption and portfolio are presented in a feedback form. She furthermore showed that the value function is a constrained viscosity solution of the associated HJB equation. Classical methods for solving HJB equations are discussed by Benton in \cite{Benton1977}. In \cite{MusielaZariphopoulou}, Musiela and Zariphopoulou applied the power-like transformation in order to linearize the non-linear PDE for the value function in the case of an exponential utility function. In the seminal paper \cite{Karatzas} Karatzas {\it et al.} investigated a similar problem of consumption-investment optimization where the problem is to maximize total expected discounted utility of consumption over time horizon $[0,T]$. For a class of utility functions, they derived explicit solutions to the HJB equation. However, in our case the aim is to maximize the expected value of the terminal utility from portfolio for a general utility function under constraints imposed on the control function and for the case of nontrivial defined contributions to the portfolio. As consequence, we have to solve the dynamic HJB equation and, in general, explicit solutions to such nonlinear PDE are no longer available.

Regarding numerical approaches for solving HJB equations associated with portfolio optimization, we can refer to finite difference methods for approximating its viscosity solution developed and analyzed by Tourin and Zariphopoulou \cite{TourinZariphopoulou}, Crandall, Ishii and Lions \cite{Crandall1992}, Nayak and Papanicolaou \cite{nayak2008}. In \cite{muthuraman2004}, Muthamaran and Sunil solved a multi-dimensional portfolio optimization problem with transaction costs. They used finite element method and iterative procedure that converts a free-boundary problem into a sequence of fixed boundary problems. In \cite{Peyrl}, Peyrl \emph{et al.} applied a successive approximation algorithm for solving the corresponding HJB equation. The fixed point-policy iteration scheme for solving discretized HJB equations is discussed in Huang \emph{et al.} \cite{Huang2010}. In \cite{Reisinger}, Witte and Reisinger presented a penalty approach for the numerical solution of discrete continuously controlled HJB equations.

In our approach we follow a different approach. Rather than solving the fully nonlinear HJB equation directly, we first transform it into a quasi-linear parabolic equation by means of the Riccati transformation. We prove existence and uniqueness of a solution to the transformed quasi-linear parabolic equation. Moreover, we derive useful bounds on the solution. These bounds can be interpreted as estimates for the coefficient of risk aversion. A special attention is put on a solution of an auxiliary parametric quadratic programming problem. It is shown that the derivative of the value function of such a convex program plays the role of a diffusion coefficient of the quasi-linear equation. Although the diffusion function need not be sufficiently smooth, we are able to prove existence, uniqueness and derive useful bounds of classical H\"older smooth solutions.

The resulting equation can be solved numerically by an iterative method based on finite volume approximation. There is an analogy to a solution of fully nonlinear generalizations of the Black-Scholes equation for pricing derivative securities (cf. \v{S}ev\v{c}ovi\v{c}, Stehl\'ikov\'a and Mikula \cite{SSM}) and the fully nonlinear HJB equation investigated in this paper. In \cite{JS} Janda\v{c}ka and \v{S}ev\v{c}ovi\v{c} suggested a numerical method for solving a fully nonlinear generalization of the Black--Scholes equation by means of its transformation to the so-called Gamma equation stated for the second derivative of the option price. In fact, the Riccati transformation is the logarithmic derivative of the derivative of the value function. Here we apply the Riccati transformation proposed and analyzed in a series of papers by Ishimura \emph{et al.} \cite{AI,IM,IsshiNaka}. In the context of a class of HJB equations with range constraints, such a transformation has been analyzed recently by Ishimura and \v{S}ev\v{c}ovi\v{c} in \cite{IshSev} where also a traveling wave solution to the HJB equation has been constructed. Concerning numerical methods for solving the transformed quasi-linear parabolic PDE there are recent papers by Ishimura, Koleva and Vulkov \cite{IKV1,IKV2,Kole,KoleVulkov} where they considered a simplified problem without inequality constraints on the optimal control function.

The paper is organized as follows. In Section 2 we formulate the problem  of our interest and the motivation behind it. Section 3 is devoted to analysis of the Riccati transformation of the HJB equation into a quasi-linear parabolic equation. The transformed function can be interpreted in terms of the coefficient of relative risk aversion of an investor. In Section 4 we analyze a class of parametric quadratic optimization problems. The goal of this section is to show that the value function is a sufficiently smooth and increasing function. Lipschitz continuity of the derivative of the value function is a crucial requirement for the proof of existence and uniqueness of a classical solution to the transformed quasi-linear parabolic equation presented in Section 5.  We also derive useful bounds of a solution to the Cauchy problem for the corresponding quasi-linear parabolic equation. Using these bounds and smoothness properties of the value function of the auxiliary parametric quadratic optimization problem, we prove existence of a classical H\"older smooth solution.  A special semi-explicit solution having the form of a traveling wave is analyzed in Section 6. Such a special solution is then utilized as a benchmark solution in Section 7, where we construct an iterative fully implicit numerical approximation scheme for solving a quasi-linear parabolic  equation.  Section 8 is devoted to application of the method to construction of an optimal response strategy for the German DAX 30 Index.

\section{Problem statement}
\label{sec:motivation}

Our motivation arises from a dynamic stochastic optimization problem in which the purpose is to maximize the conditional expected value of the terminal utility of a portfolio:
\begin{equation}
\max_{\bmtheta|_{[0,T)}} \mathbb{E}
\left[U(X_T^\bmtheta)\, \big| \, X_0^\bmtheta=x_0 \right],
\label{maxproblem}
\end{equation}
where $\{X_t^{\bmtheta}\}$ is the It\=o's stochastic process on the finite time horizon $[0,T]$, $U: \mathbb{R} \to \mathbb{R}$ is a given terminal utility function and $x_0$ a given initial state condition of  $\{X_t^{\bmtheta}\}$ at $t=0$. The function $\bmtheta:   \mathbb{R} \times [0,T) \to \R^n$ mapping $(x,t) \mapsto \bmtheta(x,t)$ represents an unknown control function governing the underlying stochastic process $\{X_t^\bmtheta\}_{t\ge0}$. Here $ \bmtheta|_{[t,T)}$ for $0\le t<T$ denotes the restriction of the control function $\bmtheta$ to the time interval $[t,T)$. We assume that $X_t^\bmtheta$ is driven by the stochastic differential equation
\begin{equation}
\d X_t^\bmtheta = \left( \varepsilon e^{-X_t} + r + \mu(\bmtheta)
-\frac12 \sigma(\bmtheta)^2 \right) \d t + \sigma(\bmtheta)
\d W_t, \label{processX}
\end{equation}
where $W_t$ denotes the standard Brownian motion and the functions $\mu(\bmtheta)$ and $\sigma(\bmtheta)$ are the drift and volatility functions depending on the control function $\bmtheta$. The parameter $\varepsilon \in \mathbb{R}$ represents a constant inflow rate of property to the system whereas $r\ge0$ is the interest rate. Many European pension systems use $\varepsilon > 0$, representing regular contribution rate to the saver's pension account as a prescribed percentage of their salary. For example, $\varepsilon=0.06-0.09$ in Slovakia, $\varepsilon=0.14$ in Bulgaria, $\varepsilon=0.02-0.05$ in Sweden (c.f. \cite{MS,KoleVulkov}).

Throughout the paper we shall assume that the control parameter $\bmtheta\in \mathcal{S}^n$ belongs to the compact simplex
\begin{equation}
\mathcal{S}^n = \{\bmtheta \in \mathbb{R}^n\  |\  \bmtheta \ge \mathbf{0}, \mathbf{1}^\transp \bmtheta = 1\} \subset \mathbb{R}^n ,
\label{admissible_set_t}
\end{equation} 
where $\mathbf{1} = (1,\cdots,1)^\transp \in \mathbb{R}^n$. It should be noted that the process $\{X^\bmtheta_t\}$ is a logarithmic transformation of a stochastic process $\{Y_t^{\tilde\bmtheta} \}_{t\ge0}$ driven by the SDE:
\begin{equation}
\d Y_t^{\tilde\bmtheta} = (\varepsilon + (r +\mu({\tilde\bmtheta})) Y_t^{\tilde\bmtheta})
\d t + \sigma(\tilde\bmtheta) Y_t^{\tilde\bmtheta} \d W_t, \label{processYeps}
\end{equation}
where $\tilde\bmtheta(y,t) = \bmtheta(x,t)$ with $x=\ln y$.

It is known from the theory of stochastic dynamic programming that
the so-called value function
\begin{equation}
V(x,t):= \sup_{  \bmtheta|_{[t,T)}} 
\mathbb{E}\left[U(X_T^\bmtheta) | X_t^\bmtheta=x \right]
\end{equation}
subject to the terminal condition $V(x,T):=U(x)$ can be used for
solving the stochastic dynamic optimization problem (\ref{maxproblem}) (cf. Bertsekas
\cite{Bertsekas}, Fleming and Soner \cite{Fleming2005} or Bardi
and Dolcetta \cite{Bardi}).  If the process $X_t^{\bmtheta}$ is driven by (\ref{processX}), then the value function $V=V(x,t)$ satisfies the Hamilton-Jacobi-Bellman (HJB) equation 
\begin{equation}
\partial_t V + \max_{ \bmtheta \in \mathcal{S}^n} \left\{
\left(\varepsilon e^{-x} + r + \mu(\bmtheta) -
\frac{1}{2}\sigma(\bmtheta)^2\right)\partial_x V + \frac{1}{2}
\sigma(\bmtheta)^2 \partial_x^2 V \right\} = 0\,,
\label{eq_HJB}
\end{equation}
for all $x\in\R,\ t\in[0,T)$ subject to the terminal condition $V(x,T):=U(x)$ (see e.g. Macov\'a and {\v S}ev{\v c}ovi{\v c} \cite{MS} or Ishimura and \v{S}ev\v{c}ovi\v{c} \cite{IshSev}).

As a typical example leading to the stochastic dynamic optimization problem (\ref{maxproblem}) in which the underlying stochastic process satisfies SDE (\ref{processX}) one can consider a problem of dynamic portfolio optimization in which the assets are labeled as $i=1,\cdots,n,$ and associated with price processes $\{Y_t^i\}_{t \ge 0}$, each of them following a geometric Brownian motion 
\[
\frac{\d Y_t^i}{Y_t^i} = \mu_i \d t + \sum_{j=1}^n \bar\sigma_{ij} \d W_t^j
\]
(cf. Merton \cite{Merton1, Merton2}, Browne \cite{BROWNE2000}, Bielecki and Pliska \cite{BPLIS2005} or Songzhe \cite{SON2006}). The value of a portfolio with weights $\tilde\bmtheta=\tilde\bmtheta(y,t)$ is denoted by $Y^{\tilde\bmtheta}_t$. It can be shown that $\{Y^{\tilde\bmtheta}_t\}_{t\ge0}$ satisfies (\ref{processYeps}). The assumption  $\bmtheta\in \mathcal{S}^n$ corresponds to the situation in which borrowing of assets is not allowed ($\theta_i\ge0$) and $\sum_{i=1}^n\theta_i = 1$. We have $\mu(\bmtheta) = \bmmu^\transp \bmtheta$ and $\sigma(\bmtheta)^2 = \bmtheta^\transp \bmSigma \bmtheta$ with $\bmmu=(\mu_1,\cdots,\mu_n)^\transp$ and  $\bmSigma = \bar{\bmSigma} \bar{\bmSigma}^\transp$ where $\bar\bmSigma=(\bar\sigma_{ij})$. The terminal function $U$ represents the predetermined terminal utility function of the investor.

\begin{remark} \label{rem:merton}
In the case of zero inflow $\varepsilon=0$, assumption (\ref{processYeps}) made on the stochastic process $\{Y_t^{\tilde\bmtheta} \}_{t\ge0}$ is related to the well-known Merton's model  for optimal consumption and portfolio selection (cf. Merton \cite{Merton1, Merton2}). However, for Merton's model, one has to consider a larger set of constraints for control function $\bmtheta$. Namely, the simplex $\mathcal{S}^n$ has to be replaced by a larger set $\mathcal{S}^n_o = \{\bmtheta \in \mathbb{R}^n\  |\  \bmtheta \ge \mathbf{0}, \mathbf{1}^\transp \bmtheta \le  1\} \subset \mathbb{R}^n$. It is worth to note that all results concerning $C^{1,1}$ smoothness of the value function $\alpha$ (see Theorem~\ref{smootheness}) as well as those regarding existence and uniqueness of classical solutions (see Theorem~\ref{existence}) and numerical discretization scheme remain true when $\mathcal{S}^n$ is replaced by $\mathcal{S}^n_o$.
\end{remark}

\section{The Riccati transformation of the HJB equation to a quasi-linear parabolic equation}
\label{sec:HJB}

Following the methodology of the Riccati transformation first proposed by Abe and Ishimura in \cite{AI} and later studied by Ishimura  \emph{et al.} \cite{IM,IsshiNaka}, Xia \cite{Xia}, or  Macov\'a and \v{S}ev\v{c}ovi\v{c} \cite{MS} for problems without inequality constraints, and further analyzed by Ishimura and \v{S}ev\v{c}ovi\v{c} \cite{IshSev}, we
introduce the following transformation:
\begin{equation}
\varphi(x,t) = 1 - \frac{\partial_x^2 V(x,t)}{\partial_x V(x,t)}.
\label{eq_varphi}
\end{equation}

\begin{remark}\label{rem:ara}
The function $a(x,t)\equiv\varphi(x,t) -1$ can be viewed as the
coefficient of absolute risk aversion for the value function
$V(x,t)$, representing the intermediate utility function of an 
investor at a time $t\in[0,T]$ (cf. Pratt \cite{P}). In the original
variable $y$, denoting $\widetilde V(y,t) = V(\ln y,t)$, we can
deduce that the function $\widetilde a(y,t) \equiv\varphi(\ln
y,t)$ is the coefficient of relative risk aversion of the
intermediate utility function  $\widetilde V(y,t)$, which is
defined as the ratio: $\widetilde a(y,t) = - y\partial_y^2
\widetilde V(y,t)/ \partial_y \widetilde V(y,t)$. 
\end{remark}

\begin{remark}\label{rem:transactioncosts}
It is worth to note that the pension saving's model model based on the SDE (\ref{processX}) takes into account neither transaction costs nor consumption. It follows from recent papers by Dai {\emph et al.} \cite{DaiSIAM, DaiJDE} that a model incorporating these effects leads to a HJB equation in two spatial dimensions. In such a case, transformation based on a simple one dimensional Riccati transformation (\ref{eq_varphi}) is not possible. 
\end{remark}

Suppose for a moment that $\varphi(x,t) >0$ for all $x\in\R$ and
$t\in[0,T]$. This assumption is clearly satisfied for $t=T$ if we
consider a function $U(x)$ which is an increasing and
concave function in the $x$ variable. We discuss more on this
assumption in Section \ref{sec:existence}. Now, problem
(\ref{eq_HJB}) can be rewritten as follows:
\begin{equation}
0 = \partial_t V + \left(\varepsilon e^{-x} + r - \alpha(\varphi)\right)\partial_x V, \qquad V(x,T):=U(x),
\label{eq_HJBtransf}
\end{equation}
where $\alpha(\varphi)$ is the value function of the following
parametric optimization problem:
\begin{equation}
\alpha(\varphi) = \min_{ \bmtheta \in \mathcal{S}^n} \{
-\mu(\bmtheta) +  \frac{\varphi}{2}\sigma(\bmtheta)^2 \}\,.
\label{eq_alpha_def}
\end{equation}

If the variance function
$\bmtheta\mapsto\sigma(\bmtheta)^2$ is strictly
convex and $\bmtheta\mapsto\mu(\bmtheta)$ linear (as discussed in Section \ref{sec:motivation}),
problem (\ref{eq_alpha_def}) belongs to a class of parametric
convex optimization problems (cf. Bank \emph{et al.}
\cite{Bank1983}).

\begin{theorem}
Suppose that the value function $V$ satisfies (\ref{eq_HJBtransf})
and the function $\varphi$ is defined as in (\ref{eq_varphi}).
Then $\varphi$ is a solution to the Cauchy problem for the
quasi-linear parabolic equation:
\begin{eqnarray}
&& \partial_t \varphi
+
\partial_{x}^2 \alpha(\varphi)
+
\partial_x
[ (\varepsilon  e^{-x} + r )\varphi + (1-\varphi) \alpha(\varphi)  ]
=0, \quad x\in\R, t\in[0,T), \label{eq_PDEphi_1} \\
&& \varphi(x,T) = 1 - U^{\prime\prime}(x)/U^\prime(x), \quad x\in\R. \nonumber
\end{eqnarray}
\end{theorem}

\begin{proof}
The statement can be easily shown  by differentiating
(\ref{eq_varphi}) with respect to $t$ and calculating derivatives
$\partial_t V$, $\partial_{x}\partial_t V$, $\partial_{x}^2
\partial_{t}V$ from (\ref{eq_HJBtransf}). Indeed, as $\partial^2_x
V = (1-\varphi)\partial_x V$, we have
\[
\partial_t\varphi =
 - \frac{\partial_{x}^2 \partial_{t}V}{\partial_x V}
+
\frac{\partial_{x}^2 V \partial_x \partial_{t} V}{(\partial_x V)^2}
=
- \frac{\partial_{x}^2 \partial_{t}V}{\partial_x V}
+ (1-\varphi) \frac{\partial_x \partial_{t} V}{\partial_x V}.
\]
Let us denote
\begin{equation}
g(x,t) = \alpha(\varphi(x,t)) - \varepsilon e^{-x} -r.
\label{Ag}
\end{equation}
Then $\partial_t V = g \partial_x V$ and therefore
\begin{eqnarray*}
\partial_x\partial_t V &=&  \partial_x g \partial_x V + g \partial^2_x V
= [\partial_x g + g (1-\varphi) ] \partial_x V,
\\
\partial^2_x\partial_t V &=&  [\partial^2_x g + \partial_x(g (1-\varphi) )
+ (\partial_x g + g (1-\varphi)) (1-\varphi) ] \partial_x V .
\end{eqnarray*}
Hence
\begin{equation}
\partial_t \varphi = - \partial_x\left( \partial_x g + g
(1-\varphi)\right)\,, \label{Bg}
\end{equation}
and $\partial_t \varphi = -\partial_x \left[
\partial_x\alpha(\varphi) + (\varepsilon  e^{-x} +r)\varphi +
\alpha(\varphi) (1-\varphi) \right]$, as claimed.

Finally, we notice that $\partial_x^2 \alpha(\varphi) = \partial_x(\alpha^{\prime}(\varphi) \partial_x \varphi)$. Moreover, if $\alpha$ is strictly increasing then \eqref{eq_PDEphi_1} indeed is a quasi-linear parabolic PDE with terminal condition at $t=T$  (see Ladyzhenskaya {\it et al.} \cite[Chapter 1, (2.4)]{LSU}). 
\end{proof}

Conversely, one can construct a solution $V(x,t)$ to the HJB
equation (\ref{eq_HJBtransf}) using a solution $\varphi$
satisfying equation (\ref{eq_PDEphi_1}). Indeed, suppose that the
function $\varphi$ satisfies (\ref{eq_PDEphi_1}). We can define a
function $V=V(x,t)$ as the unique solution to the first order
linear PDE satisfying the terminal condition:
\begin{equation}
\partial_t V - g \partial_x V =0\,, \quad V(x,T)=U(x)\,,\quad x\in\R,\ t\in[0,T),
\label{transform}
\end{equation}
where the function $g=g(x,t)$ is given by (\ref{Ag}). Let us
introduce $\psi=\psi(x,t)$ as follows:
\[
\psi=1-\frac{\partial_x^2 V}{\partial_x  V}.
\]
Then, following derivation of (\ref{Bg}) we end up with an
equation for the function $\psi$:
\[
\partial_t \psi = - \partial_x\left( \partial_x g + g (1-\psi)\right).
\]
Hence the difference $h \equiv \psi-\varphi$ satisfies a linear
PDE: $\partial_t h =  \partial_x g (h)$. Since
$\varphi(x,T)\equiv \psi(x,T)$ we deduce $\varphi(x,t)=\psi(x,t)$
for all $x\in\R$ and $t\in[0,T]$. But it means that $V$ fulfills
the fully nonlinear equation:
\begin{equation}
\partial_t V - \left[ \alpha(1- \partial_x^2 V/\partial_x V ) - \varepsilon e^{-x} -r \right] \partial_x V =0\,, \quad V(x,T)=U(x).
\label{hjb-nonlin}
\end{equation}
In other words, $V=V(x,t)$ satisfies HJB equation
(\ref{eq_HJBtransf}). Consequently, it is a solution to HJB
equation (\ref{eq_HJB}). Moreover, equation (\ref{hjb-nonlin}) is
a fully nonlinear parabolic equation which is monototone in its
principal part $\partial_x^2 V$. This way one can deduce that the
solution $V$ to (\ref{hjb-nonlin}) is unique. In summary, we have
shown that we can replace solving HJB equation (\ref{eq_HJB}) by
solving the auxiliary quasi-linear equation (\ref{eq_PDEphi_1}).

\begin{proposition}\label{equivalence}
Let $\varphi(x,t)$ be a solution to the Cauchy problem
(\ref{eq_PDEphi_1}). Then the function $V(x,t)$ given by
(\ref{transform}) is a solution to HJB equation (\ref{eq_HJB}).
Moreover, $\varphi=1- \partial_x^2 V/\partial_x  V$.
\end{proposition}

\begin{remark}
The advantage of transforming (\ref{eq_HJB}) to
(\ref{eq_HJBtransf})--(\ref{eq_alpha_def}) is that we can define
and compute the function $\alpha(\varphi)$ in advance as a result
of the underlying parametric optimization problem (either
analytically or numerically). This can be then plugged into the
quasi-linear equation (\ref{eq_PDEphi_1}) which can be solved for
$\varphi$, instead of solving the original fully nonlinear HJB
equation (\ref{eq_HJBtransf}) as well as (\ref{eq_HJB}). In this
way we do not calculate the value function $V$ itself. On the
other hand, it is only the optimal feedback strategy $\bmtheta$
which is of investor's interest and therefore $V$ is not
important, in fact. The optimal strategy $\bmtheta=\bmtheta (x,t)$ can be computed
as the unique optimal solution to the quadratic optimization problem
(\ref{eq_alpha_def}) for the parameter values
$\varphi=\varphi(x,t)$.
\end{remark}


\section{A parametric quadratic programming problem}
\label{sec:multiportf}

In the case of the example of a portfolio consisting of $n$ assets, we denote $\bmmu$ the
vector of expected asset returns and $\bmSigma$ the
covariance matrix of returns which we assume to be symmetric and positive
definite. For the portfolio return and variance we have
$\mu(\bmtheta)=\bmmu^\transp \bmtheta$ and
$\sigma(\bmtheta)^2 = \bmtheta^\transp \bmSigma\, \bmtheta$.
For $\varphi>0$, (\ref{eq_alpha_def}) becomes a problem of
parametric quadratic convex programming
\begin{equation}
\alpha(\varphi) = \min_{ \bmtheta \in \mathcal{S}^n} \{
- \bmmu^\transp \bmtheta +  \frac{\varphi}{2} \bmtheta^\transp \bmSigma\, \bmtheta \}\,
\label{eq_alpha_def_quad}
\end{equation}
over the compact convex simplex $ \mathcal{S}^n$. In this section,
we shall discuss qualitative properties of the value function
$\alpha=\alpha(\varphi)$ for this case. By $C^{k,1}(\R^+)$ we
denote the space of all functions defined on $(0,\infty)$
whose  $k$-th derivative is Lipschitz continuous. By
$\alpha^\prime(\varphi)$ we denote the derivative of $\alpha(\varphi)$
w.r. to $\varphi$.

\begin{theorem}\label{smootheness}
Let $\bmSigma \succ 0$ be positive definite and $\bmmu\in\R^n$. Then the optimal value function $\alpha(\varphi)$ defined as in (\ref{eq_alpha_def_quad}) is a $C^{1,1}$ continuous function. Moreover, $\varphi\mapsto\alpha(\varphi)$ is a  strictly increasing function and 
\begin{equation}
\alpha^\prime(\varphi) =
 \frac{1}{2} \hat\bmtheta^\transp \bmSigma
 \hat\bmtheta\,,  \label{eq_alphader_vzorec}
\end{equation}
where $\hat\bmtheta=\hat\bmtheta(\varphi) \in \mathcal{S}^n$ is
the unique minimizer of (\ref{eq_alpha_def_quad}) for
$\varphi >0$. The function $(0,\infty) \ni \varphi
\mapsto \hat{\bmtheta}(\varphi) \in\R^n$  is locally Lipschitz continuous.
\end{theorem}
\begin{proof}
First, we notice that the mapping  $(0,\infty)\ni \varphi \mapsto \hat{\bmtheta}(\varphi) \in\mathcal{S}^n$ is continuous, which can be deduced directly from basic properties of strictly convex functions minimized over the compact convex set $\mathcal{S}^n$. 

Let us denote $f(\bmtheta, \varphi):= - \bmmu^\transp \bmtheta
+ \varphi \frac{1}{2} \bmtheta^\transp \bmSigma \bmtheta$ the
objective function in problem (\ref{eq_alpha_def_quad}). Since $|\partial_{\varphi}
f(\bmtheta,\varphi)|$ is a continuous function on the compact
set $\mathcal{S}^n$, we have  $\sup_{\bmtheta \in \mathcal{S}^n}
|\partial_{\varphi} f(\bmtheta,\varphi)|= C(\varphi) <\infty$.
Strict convexity of $f$ in $\bmtheta$ implies the existence of a
unique minimizer $\hat\bmtheta\equiv \hat\bmtheta(\varphi)$  to
(\ref{eq_alpha_def_quad}). Moreover, $\partial_{\varphi}
f(\hat{\bmtheta}(\varphi),\varphi)\equiv \frac{1}{2}
\hat\bmtheta(\varphi)^\transp \bmSigma \hat\bmtheta(\varphi)$ is continuous in $\varphi$
due to continuity of $\hat{\bmtheta}(\varphi)$. Applying the
general envelope theorem due to Milgrom and Segal \cite[Theorem
2]{milgrom_segal2002} the function $\alpha(\varphi)$ is
differentiable on the set $(0,\infty)$.

Next, we prove that $\alpha^\prime(\varphi)>0$. The function
$f(\bmtheta,\varphi)$ is linear in $\varphi$ for any $\bmtheta
\in \mathcal{S}^n$. Therefore it is absolutely continuous in
$\varphi$ for any $\bmtheta$. Again, applying \cite[Theorem
2]{milgrom_segal2002}, we obtain
\begin{displaymath}
\alpha(\varphi) = \alpha(0) + \int_0^\varphi \partial_{\varphi}
f(\hat\bmtheta(\xi),\xi)\, \hbox{d}\xi\,.
\end{displaymath}
Therefore $\alpha^\prime(\varphi) =
\partial_{\varphi} f(\hat\bmtheta(\varphi),\varphi) =  \frac{1}{2} \hat\bmtheta(\varphi)^\transp \bmSigma \hat\bmtheta(\varphi)$, which is strictly positive on $\mathcal{S}^n$.  Hence $\varphi\mapsto \alpha(\varphi)$ is a $C^1$ continuous and increasing function for $\varphi>0$.

Local Lipschitz continuity of $\alpha^\prime(\varphi)$ now
follows from the general result proved by Klatte in \cite{Klatte} (see also Aubin \cite{Aubin}). Indeed, according to \cite[Theorem 2]{Klatte} the minimizer
function $\hat\bmtheta(\varphi)$ is locally Lipschitz
continuous in $\varphi$. Hence the derivative $\alpha^\prime(\varphi) =
\frac{1}{2} \hat\bmtheta(\varphi)^\transp \bmSigma
\hat\bmtheta(\varphi)$ is locally Lipschitz, as well.
\end{proof}

\begin{corollary}
Equation (\ref{eq_PDEphi_1}) is a strictly parabolic PDE, i.e.
there exist positive real numbers $\lambda^{-}, \lambda^+ \in (0,
\infty)$, such that for the diffusion coefficient
$\alpha^\prime(\varphi)$ of equation (\ref{eq_PDEphi_1}) the
following inequalities hold:
\begin{equation}
0<\lambda^- \le \alpha^\prime(\varphi) \le \lambda^+ <\infty
\qquad \hbox{ for all } \varphi >0 \,.
\label{bounds}
\end{equation}
\end{corollary}

\begin{proof}
These inequalities follow directly from
(\ref{eq_alphader_vzorec}), which is a  quadratic positive
definite form on a compact set $\mathcal{S}^n$. With regard to (\ref{eq_alphader_vzorec}) the function $\alpha^\prime(\varphi)$ attains its
maximum $\lambda^+$ and minimum $\lambda^-$.
\end{proof}

\begin{example}\label{example-dax}
An illustrative example of the value function $\alpha$ having discontinuous second derivative $\alpha^{\prime\prime}$ is based on real
market data and it is depicted in Fig.~\ref{fig:alphader_dax}. In this
example we consider the German DAX Index consisting of 30 stocks.
Based on historical data from August 2010 to April 2012 we have
computed the covariance matrix $\bmSigma$ and the vector of mean
returns $\bmmu$. One can observe that there are at least two points of
discontinuity of the second derivative
$\alpha^{\prime\prime}(\varphi)$.

\begin{figure}
\begin{center}
\includegraphics[width=0.45\textwidth]{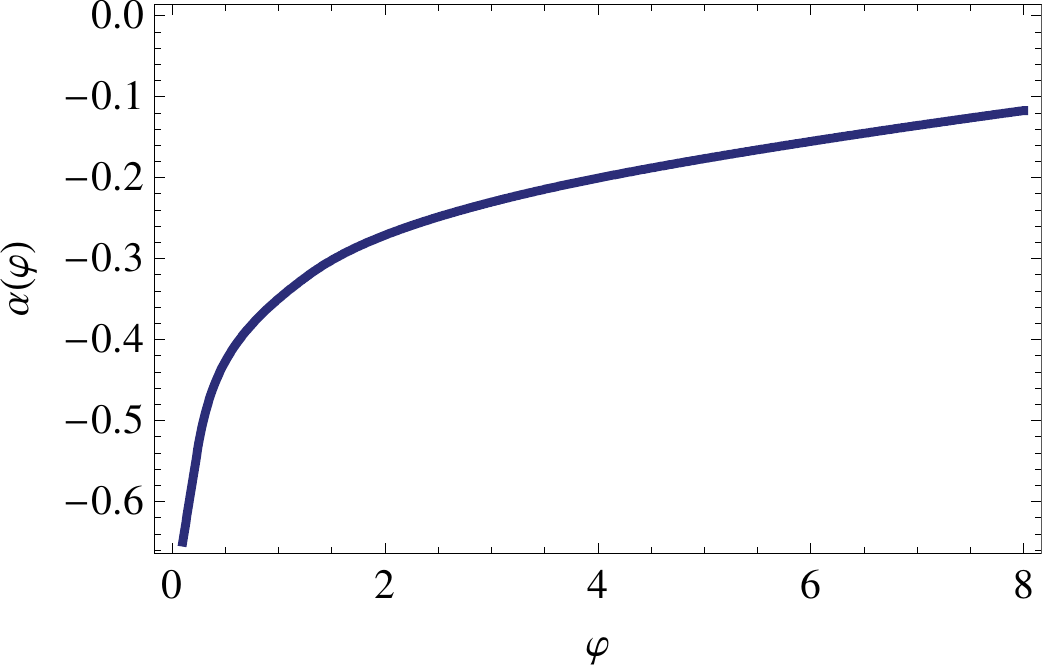}
\includegraphics[width=0.45\textwidth]{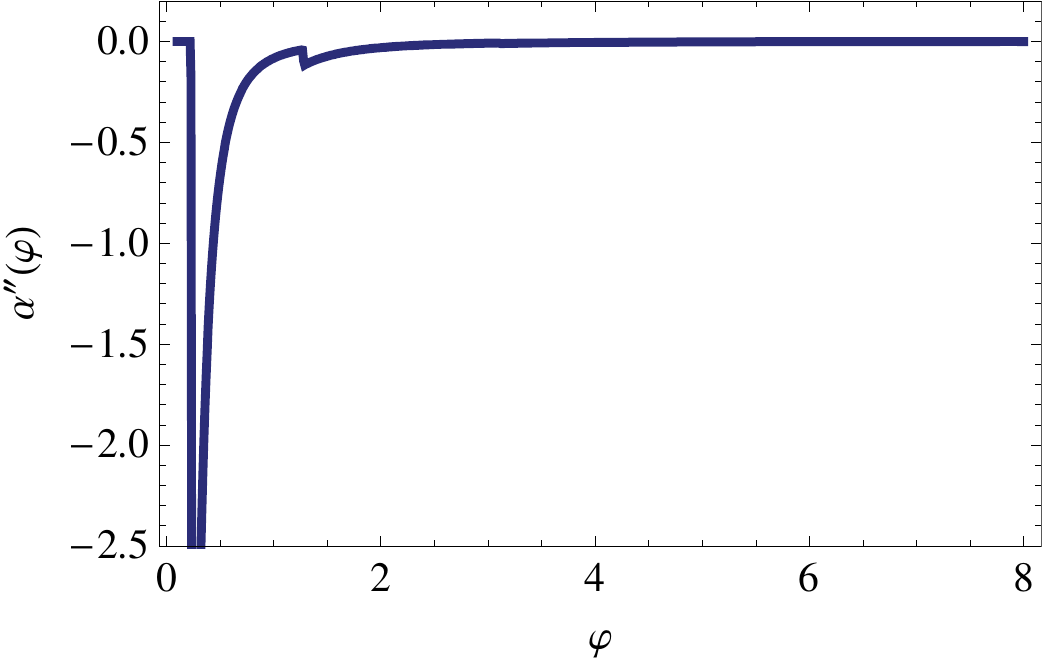}
\end{center}
\caption{%
The value function $\alpha$ and its second derivative $\alpha^{\prime\prime}$ for
the portfolio of the German DAX 30 Index, computed from historical
data, August 2010--April 2012. Source: finance.yahoo.com}
\label{fig:alphader_dax}
\end{figure}

\end{example}

\subsection{Higher smoothness of the value function.}

In this section we discuss further smoothness properties of the
value function $\alpha=\alpha(\varphi)$ in the $\varphi$
variable, for the case specified at the beginning of Section \ref{sec:multiportf}. We furthermore show that the function $\alpha$ is
locally a rational function which is concave on an open set.

Let us denote ${\mathcal I}_\emptyset$ the set
\[
{\mathcal I}_\emptyset = \{ \varphi>0\ |\  \hat\theta_i(\varphi)
>0\  \forall i=1,\cdots, n\}.
\]
Then
\[
(0,\infty) = {\mathcal I}_\emptyset \cup \bigcup_{|M|\le n-1}
{\mathcal I}_M, \ \ \hbox{where} \ \ {\mathcal I}_M = \{
\varphi>0\, |\  \hat\theta_i(\varphi) = 0 \Leftrightarrow i\in
M\},
\]
and $M$ varies over all subsets of active indices,  $M\subset
\{1,\cdots, n\}$. Here $|M|$ denotes the number of elements of the
set $M$. Since $\varphi \mapsto \hat\bmtheta(\varphi)$ is
continuous, the set $ {\mathcal I}_\emptyset$ is open.

First, let us consider the case $\varphi\in{\mathcal I}_\emptyset$. If we
introduce the Lagrange function $L(\bmtheta, \lambda) =
\frac{\varphi}{2} \bmtheta^\transp \bmSigma \bmtheta -
\bmmu^\transp \bmtheta  -\lambda \mathbf{1}^\transp \bmtheta$ then
the optimal solution $\hat\bmtheta=\hat\bmtheta(\varphi)$ and the
Lagrange multiplier $\lambda=\lambda(\varphi)$ are given by:
\[
\hat\bmtheta = \frac{1}{\varphi} \left(\bmSigma^{-1} \bmmu  + \lambda  \bmSigma^{-1} \mathbf{1}\right) , \quad \lambda = \frac{\varphi - \mathbf{1}^\transp \bmSigma^{-1} \bmmu }{ \mathbf{1}^\transp \bmSigma^{-1} \mathbf{1} }.
\]
Hence
\begin{equation}
\hat\bmtheta(\varphi) = \mathbf{a} - \frac{1}{\varphi} \mathbf{b}
\quad\hbox{and}\ \ \alpha(\varphi) = a\varphi  - \frac{b}{\varphi}
+ c, \label{A1}
\end{equation}
where $\mathbf{a}, \mathbf{b} \in\R^n$ can be expressed as follows:
\begin{equation}
\mathbf{a} =  \frac{1}{\mathbf{1}^\transp \bmSigma^{-1}
\mathbf{1}} \bmSigma^{-1} \mathbf{1},\quad \mathbf{b} =  -
\bmSigma^{-1} \bmmu  + \frac{\bmmu^\transp \bmSigma^{-1}
\mathbf{1}}{\mathbf{1}^\transp \bmSigma^{-1}
\mathbf{1}}\bmSigma^{-1} \mathbf{1}. \label{A1b}
\end{equation}
After straightforward
calculations we conclude
\begin{equation}
a= \frac{1}{2} \frac{1}{\mathbf{1}^\transp \bmSigma^{-1}
\mathbf{1}}>0, 
\ \  b= \frac{1}{2} \bmmu^\transp \bmSigma^{-1}
\bmmu - \frac{1}{2} \frac{(\mathbf{1}^\transp \bmSigma^{-1}
\bmmu)^2} {\mathbf{1}^\transp \bmSigma^{-1} \mathbf{1}}\ge 0,
\ \ c = - \frac{\mathbf{1}^\transp \bmSigma^{-1} \bmmu}
{\mathbf{1}^\transp \bmSigma^{-1} \mathbf{1}}. \label{eq:abc}
\end{equation}
The inequality $b\ge 0$ follows from the Cauchy-Schwartz inequality. Notice that $b>0$ unless the vectors $\bmmu$ and $\mathbf{1}$ are linearly dependent.

Now, if $\varphi \in \mathcal I_M$ for some subset $ M\subset
\{1,\cdots, n\}$ of active indices, then the quadratic
minimization problem (\ref{eq_alpha_def}) can be reduced to a
lower dimensional simplex ${\mathcal S}^{n-|M|}$. Hence the
function $\alpha(\varphi)$ is smooth on $\hbox{int} (\mathcal I_M
)$ and therefore $\hat\bmtheta(\varphi)$ and $\alpha(\varphi)$ are
given by:
\begin{equation}
\hat\bmtheta(\varphi) = \mathbf{a}_M - \frac{1}{\varphi} \mathbf{b}_M, \quad\hbox{and}\ \
\alpha(\varphi) = a_M\varphi  - \frac{b_M}{\varphi} + c_M,
\label{A2}
\end{equation}
for any $\varphi \in  \hbox{int} (\mathcal I_M )$ where
$\mathbf{a}_M, \mathbf{b}_M\in\R^n$ and $a_M>0, b_M\ge0$ and
$c_M\in\R$ are constants calculated using the same formulas as in
(\ref{A1b}) and (\ref{eq:abc}), where data (columns and rows) from
$\bmSigma$ and $\bmmu$ corresponding to the active indices in the
particular set $M$ are removed.

\begin{proposition}\label{pocastiachder}
The function $\varphi\mapsto
\alpha(\varphi)$ defined in (\ref{eq_alpha_def}) is a $C^\infty$
smooth function on the open set ${\mathcal J} = {\mathcal
I}_{\emptyset} \cup \bigcup_{|M|\le n-1} \hbox{\rm
int}({\mathcal I}_{M}) \subset (0,\infty)$. It is given by
(\ref{A1}) for $\varphi\in {\mathcal I}_{\emptyset}$ and by
(\ref{A2}) for $\varphi\in \hbox{\rm int}({\mathcal I}_{M})$
where $M\subset \{1,\cdots, n\}$, respectively. 
\end{proposition}

\subsection{Useful information gained from the second derivative of the value function.} \label{ex:alpha_discontin}

There is a useful information that can be extracted from the shape of $\alpha^{\prime\prime}(\varphi)$. For illustration, let us observe points of
discontinuity of $\alpha^{\prime\prime}(\varphi)$ depicted in Fig. \ref{fig:alphader_dax} for the example of  the German DAX 30 Index. The intervals between the points of discontinuities correspond to
the sets $\mathcal{I}_M$. For the portfolio of the German DAX 30 Index we obtain the sets
of active indices, corresponding to the continuity intervals as summarized in Tab. \ref{tab:activeindices}.

High values of $\varphi$ represent high risk-aversion of the investor. There is only one single asset present with a nonzero weight (equal to one) in the first
interval. This asset is the most risky one and with highest expected
return. Indeed, for lowest values of $\varphi$, investor's risk
aversion is low and therefore they do not hesitate to undergo high
risk for the sake of gaining high return.

\begin{table}
\begin{center}
\small
\begin{tabular}{l|l}
$\mathcal{I}_M$ & $M$ \\ \hline \hline
 (0, 0.23) & \{23\} \\
 (0.23, 1.27) &  \{23, 30\} \\
(1.27, 3.15) &  \{16, 23, 30\} \\
(3.15, 6.62) & \{16, 23, 27, 30\} \\
(6.62, 7.96)  &  \{16, 21, 23, 27, 30\} \\
(7.96, 8.98) & \{15, 16, 21, 23, 27, 30\} \\
(8.98, $\cdots$) & \{1, 15, 16, 21, 23, 27, 30\}  
\end{tabular} \caption{Sets of active indices for the German DAX 30 Index. The assets are labeled by 1 - Adidas, 15 - Fresenius, 16 - Fres Medical, 21 - Linde, 23 - Merck, 27 - SAP, 30~-~ Volkswagen.} \label{tab:activeindices}
\end{center}
\end{table}

Hence, if we were able to bound the parameter $\varphi$ (see Section \ref{sec:existence}) by a constant $\varphi^+ < \infty$, i.e. $\varphi(x,t) \le \varphi^+$, it would be possible to identify the intervals of continuity of $\alpha^{\prime\prime}(\varphi)$ on the interval $(0,\varphi^+]$ and the corresponding sets of active indices. This would provide the investor with information about which assets enter the portfolio with zero weight throughout the time. As will be confirmed in Section \ref{subsec:dax}, in the numerical example of the German DAX 30 Index there are only the assets from Tab. \ref{tab:activeindices}, out of the overall number of thirty, which enter the portfolio with a nonzero weight at some time from $[0,T]$; i.e. the rest of the assets stay inactive for the whole time horizon considered.

\subsection{Example: Explicit form of the value function for the 2D problem.}
\label{sec:2D}

The goal of this section is to present an explicit form of the
value function $\alpha$ for the two dimensional problem.
Furthermore, we show that the result obtained in
Theorem~\ref{smootheness} is optimal in a sense that the function
$\varphi\mapsto\alpha(\varphi)$ is only $C^{1,1}$ smooth but it
is not $C^2$ smooth. Finally, we show that in the case $n=2$ we are
able to explicitly determine the sets ${\mathcal I}_\emptyset$ and
${\mathcal I}_M$.

A vector $\bmtheta \in {\mathcal S}^2$ can be written as $\bmtheta
= (\theta, 1-\theta)^\transp$ where $\theta\in [0,1]$ is a real
number. We denote by $\mu^s,\mu^b$ the mean returns on more risky stocks and less risky
bonds and by  $\sigma^s, \sigma^b > 0$ their standard
deviations. We assume $\mu^s\ge \mu^b\ge 0$ and $\sigma^b -\varrho
\sigma^s \ge 0$ where $\varrho\in [-1,1]$ is the correlation
between returns on stocks and bonds. The mean return $\mu(\theta)$
and variance $\sigma(\theta)^2$ of the portfolio can be expressed
as
\begin{equation}
\mu(\theta) = \theta \mu^s + (1-\theta) \mu^b, \quad
\sigma(\theta)^2 =  \theta^2 \gamma - 2 \theta \delta + (\sigma^b)^2,
\end{equation}
where $\gamma=(\sigma^s)^2+(\sigma^b)^2 - 2 \sigma^s \sigma^b
\varrho$ and $\delta= (\sigma^b)^2 - \sigma^s \sigma^b \varrho$.

For a given $\varphi>0$, the objective function in
(\ref{eq_alpha_def}) is quadratic in $\theta$ with the coefficient
of the quadratic term equal to $\frac{1}{2}\gamma \varphi$.  If we
relax the inequality constraints $0\le \theta\le 1$ then it is an
easy calculus to verify that the unconstrained minimizer
$\hat{\theta}^{uc}$ is given by: $\hat{\theta}^{uc} (\varphi) =
\omega/\varphi + \delta/\gamma \ge 0$, where $\omega = (\mu^s -
\mu^b)/\gamma\ge 0$. Consequently, the optimal solution
$\hat{\theta} = \hat{\theta}(\varphi)$ for the constrained problem
over $\theta\in [0,1]$ can be written in the following form: $\hat{\theta} (\varphi) =
\min
\left\{
\omega/\varphi + \delta/\gamma,\,  1
\right\}$. 
Therefore
\begin{equation}
\alpha(\varphi) = \left\{ \begin{array}{ll}
- \mu^b - \omega\delta - \frac{\omega^2 \gamma}{2\varphi}
+
\frac{\varphi}{2}(1-\varrho^2) (\sigma^s\sigma^b)^2, & \hbox{ if }
\frac{1}{\varphi} <
\frac{1}{\omega}(1-\frac{\delta}{\gamma})\,, \\
\frac{(\sigma^s)^2}{2} \varphi - \mu^s, & \hbox{ if
}\frac{1}{\varphi} \ge
\frac{1}{\omega}(1-\frac{\delta}{\gamma})\,.
\end{array}
\right. \label{eq_alpha}
\end{equation}
In terms of the sets ${\mathcal I}_\emptyset$ and ${\mathcal I}_M$ we have $(0,\infty) = {\mathcal I}_\emptyset \cup {\mathcal I}_{\{1\}}$ where
\[
\begin{matrix}
& {\mathcal I}_\emptyset = (\omega\gamma/(\gamma-\delta),\,  \infty), \hfill
&{\mathcal I}_{\{1\}} = (0, \, \omega\gamma/(\gamma-\delta)], \hfill
&\hbox{if}\ \ \gamma >\delta, \hfill
\\
&{\mathcal I}_\emptyset = \emptyset, \hfill
&{\mathcal I}_{\{1\}} = (0, \infty), \hfill
&\hbox{if}\ \ \gamma \le \delta. \hfill
\end{matrix}
\]
With regard to Proposition~\ref{pocastiachder} the function $\varphi\mapsto \alpha(\varphi)$ is $C^{1,1}$ smooth for $\varphi>0$ and it is $C^\infty$ smooth on the set ${\mathcal J} = (0,\infty) \setminus \{\omega\gamma/(\gamma-\delta)\}$, if $\gamma > \delta$.
Notice that $\gamma > \delta$ iff $\sigma^b -\varrho \sigma^s > 0$. The latter condition is automatically satisfied for nonpositive correlation $\varrho\le 0$ between returns on stocks and bonds.

\section{Existence, uniqueness and boundedness of classical solutions}
\label{sec:existence}

In this section, we investigate properties of classical smooth
solutions to the Cauchy problem for the backward quasi-linear
parabolic equation (\ref{eq_PDEphi_1}) satisfying the terminal
condition at $t=T$. In the first part, we introduce
several function spaces we shall work with. Then we provide useful upper and lower bounds on bounded smooth solutions. Finally, following the methodology based
on the so-called Schauder's type of estimates (cf. Ladyzhenskaya
\cite{LSU}), we shall prove existence and uniqueness of classical
solutions to (\ref{eq_PDEphi_1}).

Let $\Omega=(x_L, x_R)\subset\R$ be a bounded interval. We denote $Q_T
=\Omega\times (0,T)$ the space-time cylinder. Let $0<\lambda<1$. By
$H^{\lambda}(\Omega)$ we denote the Banach space consisting of all
continuous functions $\varphi$ on $\bar\Omega$ which are $\lambda$-H\"older continuous, i.e  the H\"older semi-norm $\langle \varphi \rangle^{(\lambda)} =
\sup_{x,y\in\Omega, x\not= y} |\varphi(x) - \varphi(y)|/|x-y|^\lambda$ is finite.
The norm in the space $H^{\lambda}(\Omega)$ is then the sum of the maximum norm of $\varphi$ and the semi-norm $\langle \varphi \rangle^{(\lambda)}$. The space $H^{2+\lambda}(\Omega)$ consists of all twice continuously differentiable functions $\varphi$ in $\bar\Omega$
whose second derivative $\partial_x^2 \varphi$ belongs to $H^{\lambda}(\Omega)$. The space $H^{2+\lambda}(\R)$ consists of all functions $\varphi:\R\to\R$ such that $\varphi\in
H^{2+\lambda}(\Omega)$ for any bounded $\Omega\subset\R$.

Next, we can define the parabolic H\"older space  $H^{\lambda, \lambda/2}(Q_T)$ of functions defined on a bounded cylinder $Q_T$. It consists of all continuous functions $\varphi(x,t)$ in $\bar{Q}_T$  such that $\varphi$ is $\lambda$-H\"older continuous in the $x$-variable and it is 
$\lambda/2$-H\"older continuous in the $t$-variable. The norm is defined as the sum of the maximum norm and corresponding H\"older seminorms. The space $H^{2+\lambda, 1+\lambda/2}(Q_T)$ consists of all continuous functions on $\bar{Q}_T$ such that $\partial_t\varphi, \partial^2_x\varphi \in H^{\lambda, \lambda/2}(Q_T)$. Finally, the space $H^{2+\lambda, 1+\lambda/2}(\R\times [0,T])$ consists of all functions $\varphi:\R\times [0,T]\to \R$
such that $\varphi \in H^{2+\lambda, 1+\lambda/2}(Q_T)$ for any bounded cylinder $Q_T$. We shall also work with the Lebesgue and Sobolev spaces. By $L_p(Q_T), 1\le p\le \infty,$ we denote the Lebesgue space of all $p$-integrable functions (essentially bounded functions for $p=\infty$) defined on $Q_T$, equipped with the norm: $\Vert\varphi\Vert_{L_p} = (\int_{Q_T} |\varphi|^p)^{1/p},\  \Vert\varphi\Vert_{L_\infty} = \sup_{Q_T} |\varphi|$. The Sobolev space $W^1_2(Q_T)$ consists of all functions $\varphi\in L_2(Q_T)$ such that distributional derivatives  $\partial_x\varphi, \partial_t\varphi \in L_2(Q_T)$. The norm is defined as $\Vert\varphi\Vert_{W^1_2} = \Vert\varphi\Vert_{L_2} + \Vert\partial_t \varphi\Vert_{L_2} + \Vert\partial_x \varphi\Vert_{L_2}$. Finally, the parabolic Sobolev space $W^{2,1}_2(Q_T)$ consists of all functions $\varphi\in L_2(Q_T)$ such that 
$\partial_x\varphi, \partial^2_x\varphi, \partial_t\varphi \in L_2(Q_T)$,  $\Vert\varphi\Vert_{W^{2,1}_2} = \Vert\varphi\Vert_{L_2} + \Vert\partial_t \varphi\Vert_{L_2} 
+ \Vert\partial_x \varphi\Vert_{L_2} + \Vert\partial^2_x \varphi\Vert_{L_2}$ (cf. \cite[Chapter I]{LSU}).

We first derive lower and upper bounds of a solution $\varphi$ to the
Cauchy problem (\ref{eq_PDEphi_1}). The idea of proving upper and lower estimates for $\varphi(x,t)$ is based on construction of suitable sub- and super-solutions to the parabolic equation (\ref{eq_PDEphi_1}) (cf. \cite{PROTTER,LSU}).

\begin{remark}
Recall that the value
$\varphi(x,t) -1 $ can be interpreted as the coefficient of
absolute risk aversion for the intermediate utility (value)
function $V(x,t)$. Therefore, upper and lower
bounds for the solution $\varphi(x,t)$ can be also used in
estimation of the absolute risk aversion from above and below.
\end{remark}

\begin{proposition}
\label{th:CompPsi} Suppose that the terminal condition
$\varphi(x,T)$ is positive and uniformly bounded from above, i.e.,
there exists a constant  $\varphi^+$ such that  $0  < \varphi(x,T)
\le \varphi^+$ for any $x\in \R$. Assume $\alpha=\alpha(\varphi)$ is a smooth function   satisfying (\ref{bounds}). If $\varphi\in H^{2+\lambda, 1+\lambda/2}(\R\times [0,T])\cap L_\infty(\R\times(0,T))$, for some $0<\lambda<1$, is a bounded solution to the Cauchy problem for quasi-linear parabolic equation (\ref{eq_PDEphi_1}) then it satisfies the following inequalities:
\[
0<  \varphi(x,t) \le  \varphi^+, \quad \hbox{for any}\ t\in [0,T)\ 
\hbox{and}\ x\in\R.
\]
\end{proposition}

\begin{proof}
Equation (\ref{eq_PDEphi_1}) can be rewritten as a fully nonlinear parabolic equation of the form
\begin{equation}
\label{fullynonlinear}
\partial_\tau \varphi =\mathcal{H}(x, t,\varphi,\partial_x \varphi, \partial_x^2 \varphi),
\end{equation}
where $\tau=T-t \in (0,T)$ and $\mathcal{H} \equiv
\partial_{x}^2 \alpha(\varphi) + \partial_x \left[\alpha(\varphi) + (\varepsilon  e^{-x}+r) \varphi - \alpha(\varphi) \varphi \right]$. Notice that the right-hand side of (\ref{fullynonlinear}) is a strictly parabolic operator such that
\[
0<\lambda^- \le \partial_q \mathcal{H} (x, t,\varphi,p,q) \equiv
\alpha^\prime (\varphi) \le  \lambda_+ < \infty\,,
\]
for all $\varphi>0$. Let us define constant sub- and
super-solution $\underline{\varphi}$ and $\overline{\varphi}$ as
follows:
\[
\underline{\varphi}(x,t) \equiv 0 ,\quad  \overline{\varphi}(x,t)
\equiv \varphi^+, \quad \hbox{for all} \ x\in \R, \ t\in (0,T).
\]
Clearly, $\mathcal{H}(x, t,\underline{\varphi},\partial_x
\underline{\varphi},
\partial^2_x \underline{\varphi}) \equiv 0$, and $\mathcal{H}(x, t,\overline{\varphi},\partial_x \overline{\varphi}, \partial^2_x \overline{\varphi}) = -  (\varepsilon e^{-x}+r)  \varphi^+ <0$.
Therefore $\underline{\varphi}, \overline{\varphi}$ are indeed sub- and
super-solutions to the strictly parabolic nonlinear equation
(\ref{fullynonlinear}), i.e.
\[
\partial_\tau\underline{\varphi}
\le \mathcal{H}(t,x,\underline{\varphi},\partial_x
\underline{\varphi}, \partial^2_x \underline{\varphi}), \qquad
\partial_\tau\overline{\varphi}
\ge \mathcal{H}(t,x,\overline{\varphi},\partial_x
\overline{\varphi},
\partial^2_x \overline{\varphi}),
\]
satisfying the inequality $\underline{\varphi}(x,T) < \varphi(x,T)
\le \overline{\varphi}(x,T)$ for any $x\in \R$. The inequality $0<
\varphi(x,t) \le  \varphi^+, x\in\R, t\in (0,T),$ is therefore a
consequence of  the parabolic comparison principle for strongly
parabolic equations (see e.g. \cite[Chapter V, (8.2)]{LSU} or \cite{PROTTER}).
\end{proof}

\begin{theorem}\label{existence}
Suppose that $\bmSigma$ is positive definite, $\bmmu\in\R^n, \varepsilon,r\ge 0$ and the optimal value function $\alpha(\varphi)$ is given by
(\ref{eq_alpha_def_quad}). Assume that the terminal condition
$\varphi(x,T) = 1 - U^{\prime\prime}(x)/U^\prime(x)$, $x\in\R$, is positive and 
uniformly bounded for $x\in\R$ and belongs to the H\"older space
$H^{2+\lambda}(\R)$ for some $0<\lambda<1/2$. Then there exists a unique
classical solution $\varphi(x,t)$ to the backward quasi-linear
parabolic equation (\ref{eq_PDEphi_1}) satisfying the terminal
condition $\varphi(x,T)$. The function $t\mapsto \partial_t\varphi(x,t)$ is $\lambda/2$-H\"older continuous for all $x\in\R$ whereas $x\mapsto\partial_x\varphi(x,t)$ is Lipschitz continuous for all $t\in[0,T]$. Moreover, $\alpha(\varphi(.,.))\in H^{2+\lambda, 1+\lambda/2}(\R\times [0,T])$ and $0<\varphi(x,t) \le \sup_{x\in\R} \varphi(x,T)$ for all $(x,t)\in\R\times[0,T)$.
\end{theorem}

\begin{proof}
A key role in application of the so-called Schauder's theory on existence and uniqueness of classical H\"older smooth solutions to a quasi-linear parabolic equation is played by smoothness of its coefficients. Namely, this theory requires that the diffusion
coefficient of a quasi-linear parabolic equation is sufficiently smooth. Since
$\partial_x^2 \alpha(\varphi) = \partial_x(\alpha^\prime(\varphi) \partial_x\varphi)$ and the diffusion coefficient $\alpha^\prime(\varphi)$ is only Lipschitz continuous in $\varphi$, the backward quasi-linear parabolic equation (\ref{eq_PDEphi_1}) should be regularized first. To this end,  we construct a $\delta$-parameterized family of smooth mollifier functions $\alpha_{(\delta)}(\varphi)$ such that 
\begin{equation}
\alpha_{(\delta)}(\varphi) \rightrightarrows \alpha(\varphi), \quad\text{and}\quad 
\alpha^\prime_{(\delta)}(\varphi) \rightrightarrows \alpha^\prime(\varphi), \quad\text{as}\ \delta\to 0, 
\label{uniform}
\end{equation}
locally uniformly for $\varphi\in (0,\infty)$. Moreover, regularization can be constructed in such a way that $0<\lambda^-/2 \le \alpha_{(\delta)}^\prime(\varphi) \le 2 \lambda^+ <\infty$ for all $\varphi >0,$ and all sufficiently small $0<\delta\ll 1$. 

Now, for any $\delta>0$, by applying Theorem 8.1 and Remark 8.2 from \cite[Chapter V, pp.
495--496]{LSU} we conclude existence of a unique classical bounded solution $\varphi^\delta\in H^{2+\lambda, 1+\lambda/2}(\R\times [0,T])\cap L_\infty(\R\times(0,T))$ to the Cauchy problem 
\begin{equation}
\partial_t \varphi^\delta 
+ \partial_{x} (\alpha^\prime_{(\delta)}(\varphi^\delta) \partial_x\varphi^\delta ) 
+ \partial_x f( \cdot ,\varphi^\delta, \alpha_{(\delta)}(\varphi^\delta)) =  0, 
 \quad \varphi^\delta(x,T) = \varphi(x,T), 
\label{regularized}
\end{equation}
$x\in\R, t\in[0,T)$, where $f(x,\varphi, \alpha(\varphi)) := (\varepsilon  e^{-x} +r)\varphi + (1-\varphi) \alpha(\varphi)$. 

Let $Q_T = (x_L, x_R)\times (0,T)$ be a bounded cylinder in $\R\times  (0,T)$. By virtue of Proposition~\ref{th:CompPsi}, $\varphi^\delta$ is bounded in the norm of the space $L_\infty(Q_T)$. More precisely, 
\[
\Vert\varphi^\delta\Vert_{L_\infty(Q_T)} \le \Vert\varphi(.,T)\Vert_{L_\infty(\R)}, 
\]
for any $0<\delta\ll 1$ (see also inequality (2.31) in \cite[Chapter I]{LSU}). According to the inequality \cite[Chapter I, (6.6)]{LSU} $\varphi^\delta$ is also uniformly bounded in the space  $W^1_2(Q_T)$, i.e. there exists a constant $c_0>0$ such that 
\[
\varphi^\delta>0, \quad \Vert\varphi^\delta\Vert_{W^1_2(Q_T)} \le c_0, 
\]
for any $0<\delta\ll 1$. It means that there exists a subsequence $\varphi^{\delta_k}\rightharpoonup \varphi$ weakly converging to some element $\varphi\in W^1_2(Q_T)$ as $\delta_k\to 0$. Moreover, $\varphi^{\delta_k}(x,t) \to \varphi(x,t)$ for almost every $(x,t)$. Notice that $\varphi^{\delta_k}\to \varphi$ strongly in $L_2(Q_T)$ because of the Rellich-Kondrashov compactness theorem on the embedding $W^1_2(Q_T) \hookrightarrow L_2(Q_T)$ (cf. \cite[Chapter II, Theorem 2.1]{LSU}). 

Hence $\alpha_{(\delta_k)}(\varphi^{\delta_k}) \to \alpha(\varphi)$ and  $\alpha^\prime_{(\delta_k)}(\varphi^{\delta_k}) \to \alpha^\prime(\varphi)$ strongly in $L_2(Q_T)$. This is a consequence of the inequalities
\[
|\alpha_{(\delta)}(\varphi^\delta) - \alpha(\varphi)|
\le 
|\alpha_{(\delta)}(\varphi^\delta) - \alpha(\varphi^\delta)| 
+ 
|\alpha(\varphi^\delta) - \alpha(\varphi)| 
\le 
|\alpha_{(\delta)}(\varphi^\delta) - \alpha(\varphi^\delta)| 
+ 
\lambda^+ |\varphi^\delta - \varphi|,
\]
\[
|\alpha^\prime_{(\delta)}(\varphi^\delta) - \alpha^\prime(\varphi)|
\le 
|\alpha^\prime_{(\delta)}(\varphi^\delta) - \alpha^\prime(\varphi^\delta)| 
+ 
|\alpha^\prime(\varphi^\delta) - \alpha^\prime(\varphi)| 
\le 
|\alpha^\prime_{(\delta)}(\varphi^\delta) - \alpha^\prime(\varphi^\delta)| 
+ 
L |\varphi^\delta - \varphi|,
\]
where $L>0$ is the Lipschitz constant of the function $\varphi\mapsto \alpha^\prime(\varphi)$ (see Theorem~\ref{smootheness}) and (\ref{uniform}) ). 

Multiplying equation (\ref{regularized}) by a function $\eta\in W^1_2(Q_T)$ vanishing on the boundary $\partial Q_T$ and integrating it over the domain $Q_T$ yields  the integral identity:
\[
\int_{Q_T} \partial_t\varphi^\delta \: \eta \:\d x\d t - \int_{Q_T} \left(
\alpha^\prime_{(\delta)}(\varphi)\: \partial_x\varphi^\delta + f(x,\varphi^\delta, \alpha_{(\delta)}(\varphi^\delta)) \right) \: \partial_x \eta\:\d x\d t = 0. 
\]
Passing to the limit $\delta_k\to 0$ we conclude that $\varphi\in W^1_2(Q_T)$ is a weak solution to the backward quasi-linear parabolic equation (\ref{eq_PDEphi_1}) satisfying the integral identity
\[
\int_{Q_T} \partial_t\varphi \: \eta \:\d x\d t - \int_{Q_T} \left(
\alpha^\prime(\varphi)\: \partial_x\varphi + f(x,\varphi, \alpha(\varphi)) \right) \: \partial_x \eta\:\d x\d t = 0 
\]
for any $\eta\in W^1_2(Q_T)$ vanishing on the boundary $\partial Q_T$. Since 
\begin{equation}
\partial_t\varphi + \partial^2_x \alpha(\varphi) + \partial_x f =0
\label{rovnica}
\end{equation}
and $\varphi, f\in W^1_2(Q_T)$ we have $\partial^2_x \alpha(\varphi) \in L_2(Q_T)$. Furthermore, $\partial_t \alpha(\varphi) \in L_2(Q_T)$ because $\varphi\mapsto\alpha^\prime (\varphi)$ is Lipschitz continuous (see Theorem~\ref{smootheness}), $\alpha^\prime(\varphi)>\lambda^-$  and $\partial_t\varphi\in L_2(Q_T)$. Hence $\alpha(\varphi) \in W^{2,1}_2(Q_T)$. 

Recall that the parabolic Sobolev space $W^{2,1}_2(Q_T)$ is continuously embedded into the H\"older space 
$H^{\lambda, \lambda/2}(Q_T)$ for any $0<\lambda<1/2$ (cf. \cite[Lemma 3.3, Chapter II]{LSU}). It follows from equation (\ref{rovnica}) that the transformed function $z(x,t) := \alpha(\varphi(x,t))$ is a solution to the quasi-linear parabolic equation in the non-divergent form: 
\[
\partial_t z + \zeta(z) \left[ 
\partial^2_x z 
+ \partial_x f(x, \beta(z), z)
\right] =0, \qquad z(x,T) = \alpha(\varphi(x,T)),
\]
where $\zeta(z) = \alpha^\prime(\beta(z))$ and $z\mapsto \beta(z)$ is the inverse function to the increasing function $\varphi\mapsto \alpha(\varphi)$, i.e. $\alpha(\beta(z)) = z$ for any $z$. Clearly, $z\mapsto \beta(z), \beta^\prime(z)$ are Lipschitz continuous and so $z\mapsto \zeta(z)$ is Lipschitz continuous as well. Next we make use of a simple boot-strap argument to show that $z=z(x,t)$ is sufficiently smooth. Clearly, it is a solution to the linear parabolic equation in non-divergence form
\[
\partial_t z + a(x,t) \partial^2_x z  + b(x,t) \partial_x z = F(x,t), \qquad z(x,T) = \alpha(\varphi(x,T)),
\]
where $a(x,t):=\zeta(z), b(x,t) = \zeta(z) \left( (\varepsilon e^{-x} +r) \beta^\prime(z) +1 -\beta(z) - z \beta^\prime(z) \right)$ and $F(x,t)=(\varepsilon e^{-x} +r)\beta(z)$ with $z=z(x,t)$. All the coefficients $a, b, F$ belong to the H\"older space $H^{\lambda,\lambda/2}(Q_T)$ because $z\in H^{\lambda, \lambda/2}(Q_T)$. With regard to \cite[Theorem 12.2, Chapter III]{LSU} we have $z\in H^{2+\lambda, 1+\lambda/2}(Q_T)$ and the proof of theorem easily follows.

\end{proof}

\begin{remark}\label{rem:arafunct}
Let us consider a utility function $U(x)= - \frac{1}{a-1} \exp(-(a-1) x)$  which represents an investor with constant coefficient $a>1$ of absolute risk aversion.
Then for the terminal condition $\varphi(x,T)$ we have
$\varphi(x,T)\equiv a$ is a constant function fulfilling all
assumptions of Theorem~\ref{existence} made on the terminal
function $\varphi(.,T)$.
\end{remark}

\begin{remark}\label{rem:generalalpha}
It follows from the proof of Theorem~\ref{existence} that its statement on existence of a H\"older smooth solution $\varphi$ to (\ref{eq_PDEphi_1}) remains true when the value function $\alpha(\varphi)$ is a general $C^{1,1}$ smooth function satisfying the estimates (\ref{bounds}). This allows for consideration of a broader class of value functions defined as in (\ref{eq_alpha_def}) (see also Remark~\ref{rem:merton}).
\end{remark}

Combining Theorems~\ref{smootheness} and \ref{existence} we obtain the following corollary:

\begin{corollary}\label{cor:continuity}
Under the assumptions of Theorem~\ref{existence} there exists a
unique continuous optimal response function $\bmtheta=
\bmtheta(x,t)$ to HJB equation (\ref{eq_HJB}). It is given by
$\bmtheta(x,t) = \hat\bmtheta(\varphi(x,t))$ where $
\hat\bmtheta(\varphi)$ is the optimal solution to
(\ref{eq_alpha_def_quad}) for $\varphi=\varphi(x,t)$. The function
$\R\ni x \mapsto \bmtheta(x,t)\in\R^n$ is Lipschitz continuous for all $t\in[0,T]$.
\end{corollary}

\section{A traveling wave solution}
\label{sect:TW}

The aim of this section is to construct a semi-explicit traveling
wave solution to quasi-linear equation (\ref{eq_PDEphi_1}). We
shall utilize such a special solution for testing purposes of the
numerical accuracy and estimating the convergence rate of the
numerical scheme proposed in Section \ref{sec:numscheme}. In order
to construct a traveling wave solution we shall assume
$\varepsilon=0, r=0$  and $\bmSigma$ is positive definite. In this case
\begin{equation}
\partial_t \varphi
+
\partial_{x}^2 \alpha(\varphi)
+
\partial_x
\left[ \alpha(\varphi) - \alpha(\varphi) \varphi \right] =0, \quad
x\in\R, t\in[0,T). \label{EQUX}
\end{equation}
In Theorem~\ref{smootheness} we showed that the function
$\alpha(\varphi)$ is a strictly increasing and locally $C^{1,1}$
smooth function in $\varphi$. Following the analysis and ideas due
to Ishimura and \v{S}ev\v{c}ovi\v{c} (cf. \cite{IshSev}) we shall
construct a traveling wave solution to (\ref{EQUX}) of the form
\[
\varphi(x,t) = v(x + c(T-t)), \qquad x\in\R, \ t\in [0,T],
\]
with the wave speed $ c \in \R$ and the wave profile $v= v(\xi)$.
We notice that the terminal condition $\varphi(x,T)$ to
(\ref{EQUX}) is just the traveling wave profile $v(x)$. 

\begin{remark}
In terms
of the coefficient of absolute risk aversion
$a(x)=-U^{\prime\prime}(x)/U^\prime(x)$ we have $a(x)=v(x)-1$.
Hence, a decreasing traveling wave profile corresponds to a
utility function with decreasing coefficient of absolute risk
aversion $a(x)$. It might be therefore associated with an investor
having higher risk preferences with increasing volume of the
portfolio value $x$.
\end{remark}

Inserting $\varphi(x,t) = v(x + c(T-t))$ into (\ref{EQUX}) we
deduce existence of a constant $K_0\in\R$ such that
\[
\frac{\hbox{d}}{\hbox{d}\xi} \alpha(v(\xi)) = G(v(\xi)), \quad
\hbox{where}\ \ G(v) =  K_0 +  c v  -  \alpha(v) (1-v),
\]
for any $\xi\in\R$. Let us define a new auxiliary variable $z =
\alpha(v)$. Then the function $z=z(\xi)$ satisfies the ODE:
\begin{equation}
z^\prime(\xi) = F (z(\xi)),\quad \xi\in\R,
\label{ode-riskseekingX}
\end{equation}
where $F(z) =  G(\alpha^{-1}(z)) = K_0 + c \alpha^{-1}(z)  - z +   z \alpha^{-1}(z)$.

Now, let us prescribe arbitrary limiting values $0<v^-< v^+<
\infty$ for the traveling wave profile $v(\xi)$ corresponding to
the limits $v^-=\lim_{\xi\to \infty} v(\xi)$, $v^+=\lim_{\xi\to
-\infty} v(\xi)$. We denote by $z^\pm$ the corresponding
$z$-values, i.~e. $z^\pm = \alpha(v^\pm)$. Thus $v^\pm$ are roots
of the function $G$, $G(v^\pm) =0$. Consequently, $F(z^\pm) =0$.

Given $0<v^-< v^+$, the traveling wave speed $c$ and the intercept $K_0$ are uniquely determined from the equation $G(v^\pm)=0$, i.e.
\begin{equation}
c= \frac{\alpha(v^+)(1-v^+) - \alpha(v^-)(1-v^-) }{v^+-v^-}, \qquad K_0= -c v^+ + \alpha(v^+)(1-v^+).
\label{travlspeedcK0}
\end{equation}

According to  Proposition~\ref{pocastiachder}, for any $v\in
{\mathcal J}\subseteq(0,\infty)$, the function $v\mapsto
\alpha(v)$ is $C^\infty$ smooth and it has a form of $\alpha(v) =
a v - b/v  + c$ for some constants $a>0,b\ge 0$ and $c\in\R$. As a
consequence we obtain $h^{\prime\prime}(v)  = -2a -2b/v^3 <0$
where $h(v) := \alpha(v) (1-v)$. Assume $v^\pm \in {\mathcal J}$.
Since $G^\prime(v) = (h(v^+) - h(v^-))/(v^+ - v^-) - h^\prime(v)$
and $h^{\prime\prime}(v^\pm) <0$ we obtain $G^\prime(v^-) < 0,
G^\prime(v^+) > 0$ and $G(v) < 0$ iff $v \in (v^-, v^+)$. In
Fig.~\ref{fig:Gv} we plot the function $G(v)$ calculated from the
function $\alpha$ corresponding to the case of the German DAX 30
Index (see Fig.~\ref{fig:alphader_dax}  and
Example~\ref{example-dax}). We prescribed the roots:  $v^-=0.3$
and $v^+=1.5$.

\begin{figure}
\begin{center}
\includegraphics[width=0.35\textwidth]{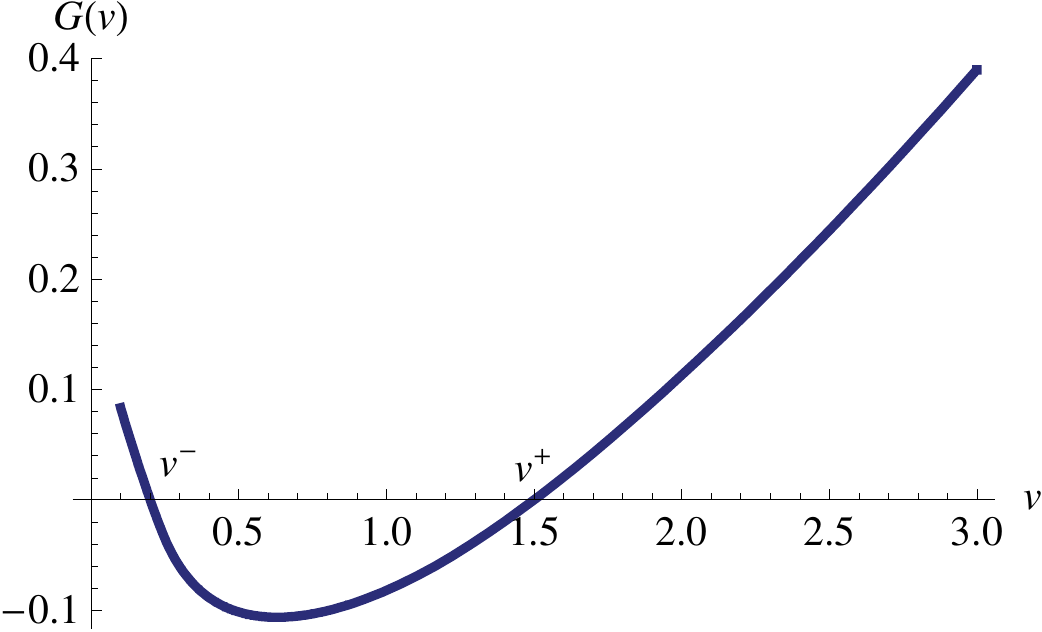}
\end{center}
\caption{%
The function $G(v)$ calculated from $\alpha$ corresponding to the
case of the German DAX 30 Index. Its roots were prescribed as
$v^-=0.3$ and $v^+=1.5$.} \label{fig:Gv}
\end{figure}

Since $F(z) = G(\alpha^{-1}(z))$ and the function $\alpha$ is
increasing we obtain $F^\prime(z^-) <0$ and  $F^\prime(z^+) >0$.
Hence $z^-$ is a stable and $z^+$ an unstable stationary solution
to (\ref{ode-riskseekingX}), i.e. $\lim_{\xi\to\pm\infty} z(\xi)=
z^\mp$ for any solution $z(\xi)$ to (\ref{ode-riskseekingX}) such
that $z(0)\in (z^-, z^+)$.

\begin{theorem}\label{theorem:riskseeking}
Assume $v^\pm \in {\mathcal J}$ are two limiting values $0<v^-<
v^+$. Up to a shift in the $x$ variable there exists a unique
traveling wave solution $\varphi(x,t) = v(x + c (T-t))$ such that
$\lim_{x\to-\infty}\varphi(x,t) = v^+$ and $\lim_{x\to \infty}
\varphi(x,t) = v^-$. The traveling wave profile $v(\xi)$ is a
decreasing function given by $v(\xi) = \alpha^{-1}(z(\xi))$ where
$z=z(\xi)$ is a solution to the ODE  (\ref{ode-riskseekingX}). The
traveling wave speed $c\in\R$ is given by (\ref{travlspeedcK0}).

\end{theorem}

\section{A numerical finite volume approximation scheme}
\label{sec:numscheme}

This section is devoted to construction of a numerical
approximation scheme for solving the Cauchy problem for the
quasi-linear parabolic equation (\ref{eq_PDEphi_1}). Recall that,
instead of solving the fully nonlinear HJB equation (\ref{eq_HJB})
containing the maximal operator, we proposed its transformation to
the quasi-linear parabolic equation (\ref{eq_PDEphi_1}). In
construction of the iterative numerical scheme we follow the
method of a finite volume approximation scheme (cf. LeVeque
\cite{LeV}) combined with a nonlinear equation iterative solver
proposed by Mikula and K\'utik in \cite{KutikMikula}. There they applied the iterative
finite volume method for solving the generalized Black-Scholes
equation with a volatility term nonlinearly depending on the
second derivative of the option price.

Equation (\ref{eq_PDEphi_1}) belongs to a subclass of quasi-linear
parabolic equations of the general form:
\begin{equation}
\partial_t \varphi + \partial_{x}^2 A(\varphi, x,t) + \partial_{x} B(\varphi,
x,t)+ C(\alpha,\varphi,x,t) = 0, \quad x\in\R, t\in[0,T),
\label{eq_numerics_general}
\end{equation}
satisfying the terminal condition at $t=T$ (cf. \cite{LSU}). In our model we have
\[
A(\varphi,x,t)=\alpha(\varphi),\quad  B(\varphi,x,t)=  (\varepsilon
e^{-x} +r)\varphi + \alpha(\varphi) (1-\varphi), \quad  C \equiv 0.
\]

In order to keep standard PDE notation, we transform the equation
from backward time to a forward one via $\widetilde\varphi(x,
\tau) := \varphi(x, T-t)$. Subsequently, we obtain
$\partial_{\tau} \widetilde\varphi = -
\partial_t \varphi$ and therefore
\begin{equation}
\partial_\tau \widetilde\varphi
=
 \partial_{x}^2 \widetilde A(\widetilde\varphi, x,\tau) + \partial_{x} \widetilde B(\widetilde\varphi,
x,\tau)+ \widetilde C(\alpha,\widetilde\varphi,x,\tau),
\quad\hbox{for any}\ x\in\R, \tau\in(0,T], \label{eq_num}
\end{equation}
with an initial condition
$\widetilde\varphi(x,0)=\widetilde\varphi_0(x) \equiv
\varphi(x,T)$, where $\widetilde A(\widetilde\varphi, x,\tau)
\equiv A(\varphi, x,T-\tau)$  is increasing in $\varphi$, and
$\widetilde B(\widetilde\varphi, x,\tau) \equiv B(\varphi,
x,T-\tau), \widetilde C(\alpha,\widetilde\varphi,x,\tau) \equiv
C(\alpha,\varphi,x,T-\tau)$. For convenience, we shall
drop the $\widetilde{\phantom{a}}$ sign in the following, but we
shall keep in mind that we work with the transformed functions
instead.

Let us consider a bounded computational domain $[x_L, x_R]$ and
spatial discretization mesh points $x_i = x_L+ i h$ for
$i=0,\cdots,n+1$ where $h=(x_R-x_L)/(n+1)$. So $x_0 =x_L$ and
$x_{n+1}=x_R$. The inner mesh points $x_i$, $i=1,\cdots,n,$ are the
centers of the finite volumes cells $(x_{i - \frac{1}{2}}, x_{i
+\frac{1}{2}})$, for simplicity denoted as $(x_{i-}, x_{i+})$. We
have $h=x_{i+} - x_{i-}$. Let us denote $\tau^j=j k, j=0, \cdots,
m$ the time steps, $k=T/m$. Integrating equation (\ref{eq_num})
over finite volumes, applying the midpoint rule on the left-hand
side integral and approximating the time derivative by forward
finite difference with step $k$, we end up with a set of equations
\begin{equation}
\varphi_i^{j+1}  = \frac{k}{h}(I_1 + I_2) + \varphi_i^j\,, \quad
i=1, \cdots ,n,\ j=0, \cdots, m,
\end{equation}
where we have denoted
\begin{equation}
I_1=\int_{x_{i-}}^{x_{i+}} \partial_x (\partial_x A(\varphi,
x,\tau) + B(\varphi, x,\tau)) {\d}x, \quad
 I_2=\int_{x_{i-}}^{x_{i+}} C(\alpha,\varphi, x,\tau) {\d}x\,. \label{eq_defI1}
\end{equation}
Depending on whether the above integrals are being computed on the
$j$-th or the $(j+1)$-th layer, we obtain different approximations.
The symbol $^\star$ will stand either for $j$ or $j+1$.

In order to compute the integral $I_2$ we apply the midpoint rule. We obtain
\begin{equation}
I_2^\star= h C(\alpha_i^\star,\varphi_i^\star, x_i,\tau^\star)\,.
\label{eq_defI3}
\end{equation}
Concerning the integral $I_1$, we shall use the following
notation:
\begin{eqnarray*}
&& D_{i\pm}^\star =
\partial_{\varphi} A(\varphi,x,\tau) |_{\varphi_{i
\pm}^\star, {x_{i \pm}}, \tau^\star}, \quad E_{i\pm}^\star =
\partial_{x} A(\varphi,x,\tau) |_{\varphi_{i \pm}^\star,, {x_{i \pm}},
\tau^\star},
\\
&& F_{i\pm}^\star = B(\varphi, x,\tau)|_{\varphi_{i \pm}^\star,
{x_{i \pm}}, \tau^\star}, \quad
\partial_x \varphi|_{i\pm}^\star = \partial_x \varphi(x,\tau)|_{{x_{i \pm}},
\tau^\star}\,.
\end{eqnarray*}
Using central spatial differences we obtain the following
numerical scheme for solving the general equation (\ref{eq_num}):
\begin{equation}
\varphi_i^{j+1}  = \frac{k}{h}(D_{i+}^\star \partial_x
\varphi|_{i+}^\star - D_{i-}^\star \partial_x \varphi|_{i-}^\star
+ E_{i+}^\star - E_{i-}^\star + F_{i+}^\star - F_{i-}^\star +
I_2^\star) + \varphi_i^j \label{eq:num}
\end{equation}
for $i=1,\cdots,n$, with approximation of the derivatives
\[
\partial _x \varphi|_{i+}^\star \approx
\frac{\varphi(x_{i+1},\tau^\star)-\varphi(x_i,\tau^\star)}{h},
\quad
\partial _x \varphi|_{i-}^\star \approx
\frac{\varphi(x_i,\tau^\star)-\varphi(x_{i-1},\tau^\star)}{h}\,.
\]
We shall pay our attention to the boundary values at $x_{0}$ and
$x_{n+1}$ later.

\smallskip
\noindent {\bf A simplified semi-implicit scheme.} To compute a solution at the new time layer $j+1$, we take the terms $D_{i\pm}^\star, E_{i\pm}^\star, F_{i\pm}^\star$ from the previous time layer with $\star=j$ and the term $\partial_x \varphi|_{i\pm}^\star$ from the new layer with $\star=j+1$. Reorganizing the new layer terms to the left-hand side and the old-layer terms to the right-hand side, we arrive at
\begin{eqnarray*}
-\frac{k}{h^2}D_{+} \varphi_{i+1}^{j+1}
&+& (1+\frac{k}{h^2}(D_{i+}^j+D_{i-}^j)) \varphi_i^{j+1} -
\frac{k}{h^2} D_{i-}^j \varphi_{i-1}^{j+1} \\
&=& \frac{k}{h^2}(I_2^j + E_{i+}^j- E_{i-}^j + F_{i+}^j - F_{i-}^j)
+ \varphi_i^j\,,
\end{eqnarray*}
which is a tridiagonal system which can be effectively solved by the Thomas algorithm.

\smallskip
\noindent {\bf An iterative fully implicit scheme.} We
take $^\star = j+1$ in all terms of (\ref{eq:num}) and
$\varphi_i^{j+1}$ will be computed iteratively as follows: we
denote $r_i^l$ the $l$-th iterative approximation of
$\varphi_i^{j+1}$, $i=1,\cdots,n$, starting with $r_i^0 :=
\varphi_i^j$. In each iterate we solve the tridiagonal system for $r_i^{l+1}$, $i=1,\cdots,n$,
with the nonlinear terms $I_2^{\star,l}, D_{i\pm}^{\star,l},
E_{i\pm}^{\star,l}, F_{i\pm}^{\star,l}$ evaluated at
$\tau^\star=\tau^{j+1}$ and $\varphi_i^{j+1} \approx r_i^{l}$. We
update $r_i^{l} := r_i^{l+1}$ until an accuracy criterion is met
and then we put $\varphi_i^{j+1}:=r_i^{l}$ from the last iterate.

\smallskip
\noindent {\bf Boundary conditions.} We consider two classes of boundary conditions: inhomogeneous Dirichlet, and mixed Robin type of homogeneous b.c.:
\begin{equation}
\begin{array}{ll}
\hbox{Dirichlet b.c.} & \varphi(x_L,t) = \varphi_L(t),
\varphi(x_R,t) = \varphi_R(t)\,,
\\
\hbox{Robin b.c.}     & \partial_x\varphi (x,t)
= d \varphi(x,t) \ \hbox{at}\ x=x_L, x_R\,,
\end{array}
\end{equation}
where the boundary functions $\varphi_L(t), \varphi_R(t)$ are
prescribed for the Dirichlet b.c., and $d \in\R$ is constant for
the Robin type of b.c. After discretization and using finite
differences, we obtain the discrete b.c.:
\[
\varphi_{0}^j = L \varphi_L^j + M \varphi_1^j,
\quad
\varphi_{n+1}^j = R \varphi_R^j + N \varphi_n^j\,,
\]
where $L=R=1, M=N=0$ for the case of Dirichlet b.c., and  $L=R=0, M=N=1/(1+d h)$ for the mixed Robin type of boundary conditions.

In our numerical approximation of the quasilinear parabolic equation (\ref{eq_PDEphi_1}) we use the following boundary conditions:
\begin{equation}
\partial_x \varphi(x,t) - \varphi(x,t) =0, \quad \text{at}\ x=x_L, \qquad 
\partial_x \varphi(x,t)  =0, \quad \text{at}\ x=x_R, 
\end{equation}
for all $t\in[0,T]$. The boundary condition at $x=x_L$ is based on the following reasoning: if $\varepsilon >0$ then, in the limit $x\to-\infty$, the dominant term in the equation 
$\partial_t \varphi + \partial_{x}^2 \alpha(\varphi) +
\partial_x
[ (\varepsilon  e^{-x}+r) \varphi + (1-\varphi) \alpha(\varphi)  ] =0$ is equal to $\partial_x
[ (\varepsilon  e^{-x}+r) \varphi(x,t) ]$. To balance this term one has to assume $\lim_{x\to-\infty} \partial_x(e^{-x} \varphi(x,t)) =0$. It means that $\lim_{x\to-\infty} \partial_x \varphi(x,t) - \varphi(x,t) =0$. The right boundary condition follows from the fact that, in the limit $x\to\infty$, equation (\ref{eq_PDEphi_1}) becomes $\partial_t \varphi + \partial_{x}^2 \alpha(\varphi) + \partial_x [ r\varphi + (1-\varphi) \alpha(\varphi)  ] =0$ having a constant solution and so $\lim_{x\to+\infty} \partial_x\varphi(x,t)=0$.

\subsection{Numerical benchmark to a traveling wave solution.}

We test the accuracy of the implicit scheme described above, using
the traveling wave analytical solution as described in Section
\ref{sect:TW} for the German DAX 30 Index and for $\varepsilon=0, r=0$.
We consider the time horizon $T=10$ and the computational domain
$[x_L, x_R] = [-4,4]$. In order to compute the semi-analytical
traveling wave solution $\varphi(x,t)$, we choose the limiting
values $v^-=0.3$, $v^+=1.5$. We solve equation
(\ref{ode-riskseekingX}) by means of the Merson method
(Runge-Kutta method of the 4th order) over the interval $[x_L, x_R
+c T]$. In the numerical scheme we use Dirichlet boundary
conditions on both ends, with values taken from the
semi-analytical traveling wave solution. For clarification, we
compute the function $\alpha(\varphi)$ numerically using the
Matlab function {\it quadprog}, with a very fine discretization
(of the order $10^{-5}$) of the considered domain of $\varphi$,
and so we consider it exact enough to substitute the exact
analytical solution. Having computed $\alpha(\varphi)$, we proceed
with solution of the quasi-linear PDE (\ref{eq_PDEphi_1}) by
means of the iterative implicit finite volume numerical scheme. As
the stopping criterion for the microiterates we choose the
$L_{\infty}$ norm of the difference of two consecutive iterates to
be less than tolerance $tol=10^{-9}$. We solve equation
(\ref{ode-riskseekingX}) using the embedded Matlab function {\it
ode45} with relative tolerance set to $10^{-8}$.

\begin{table}
\begin{center}
\small
\begin{tabular}{l | l | l|| l | l }
$h$ & $L_\infty((0,T): L_2)$-err & {$EOC_{k=0.1h}$} & $L_2((0,T):
W^1_2)$-err & $EOC_{k=0.1h}$
\\ \hline \hline 0.1 & 0.92313e-03& {-} & 1.19224e-03 & {-}
\\ 0.05 & 0.46046e-03 & 1.003
 & 0.68451e-03 & 0.801
 \\ 0.025 & 0.23194e-03 &
0.989
 & 0.38057e-03 & 0.847
\\
0.0125 & 0.11867e-03 & 0.967 & 0.20687e-03 & 0.879
 \\ 0.00625 & 0.06004e-03  &  0.983 &
0.11737e-03 & 0.818
\end{tabular}
\end{center}
\caption{%
The $L_\infty((0,T): L_2)$ and $L_2((0,T): W^1_2)$ norm of the error
of the numerical solution with the spatial step $h$ and time-space
step binding $k=0.1h$ and the exact traveling wave solution. The
experimental order of convergence.} \label{tab:EOC1}
\end{table}

\begin{table}
\begin{center}
\small
\begin{tabular}{l | l | l|| l | l }
$h$ & $L_\infty((0,T): L_2)$-err & {$EOC_{k=10h^2}$} & $L_2((0,T):
W^1_2)$-err & $EOC_{k=10h^2}$
\\ \hline \hline 0.1 & 9.47564e-03 & {-} & 14.51654e-03 & {-}
\\ 0.05 & 2.38427e-03 & 1.991
 & 3.84091e-03 & 1.918
 \\ 0.025 & 0.59656e-03 &
1.999
 & 0.98843e-03 & 1.958
\\
0.0125 & 0.14907e-03 & 2.001
 & 0.25677e-03 &
1.945
 \\ 0.00625 & 0.03725e-03  &  2.001
 &
0.08456e-03 & 1.602
\end{tabular}
\end{center}
\caption{
\small
The $L_\infty((0,T): L_2)$ and $L_2((0,T): W^1_2)$ norm of the error
of the numerical solution with the spatial step $h$ and time-space
step binding $k=10h^2$ and the exact traveling wave solution. The
experimental order of convergence.} \label{tab:EOC2}
\end{table}

\begin{figure}
\begin{center}
\includegraphics[width=0.35\textwidth]{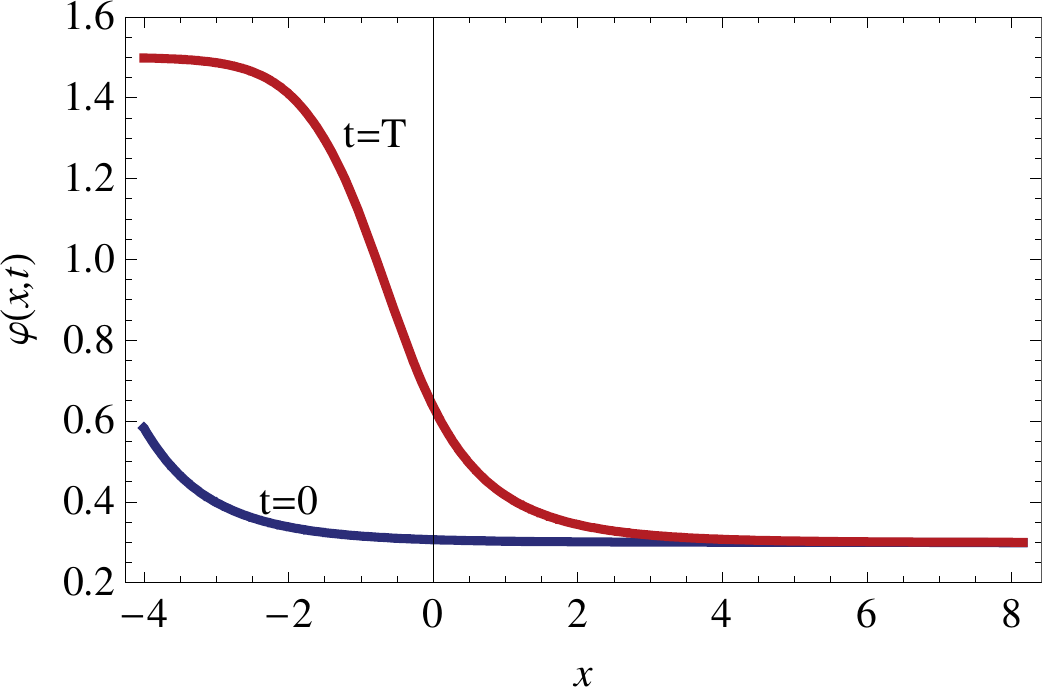}
\end{center}
\caption{
A traveling wave solution $\varphi(x,t)$ for $t=0$ and $t=T$.}
\label{fig:TWprofile}
\end{figure}

Tab.~\ref{tab:EOC1} indicates that the scheme is empirically of
the first order accurate in the $L_\infty((0,T): L_2)$ and
$L_2((0,T): W^1_2)$ norms when we restrict the time step $k$ by
$k=0.1h$. It is of the second order of convergence when $k=10h^2$,
see Tab.~\ref{tab:EOC2}. The so-called experimental order of
convergence (EOC) corresponds to the order $r>0$ of convergence
such that $err(h)=O(h^r)$ where $err(h)$ is the norm of the
difference of the numerical solution with the spatial step $h$ and
the exact traveling wave solution, i.e.
\[
r_i=\frac{\ln(err_{i}/err_{i-1})}{\ln(h_i/h_{i-1})}.
\]
Fig.~\ref{fig:TWprofile} depicts the analytical traveling wave
profile for times $t=0$ and $t=T$.

\section{Application to portfolio optimization}
\label{subsec:dax}

In this section we present an example in which our goal is to optimize a portfolio
consisting of $n=30$ assets of the German DAX 30 Index. The regular contribution to the portfolio is set to $\varepsilon=1$ and $r=0$. We consider the utility
function of the form
\begin{equation}
U(x)= - \frac{1}{a-1} \exp(-(a-1) x)\,, \label{eq:ara2}
\end{equation}
where we set the coefficient
of absolute risk aversion $a=9$. Notice that the constant absolute
risk aversion (CARA) utility function (\ref{eq:ara2}) corresponds
to the constant relative risk aversion (CRRA) function $\widetilde
U(y) = - \frac{1}{a-1} y^{-a+1}$ when expressed in the variable
$y= e^x$. We consider the finite time horizon $T=10$. Our
guess about the minimal and maximal possible values of $y$ is $y_L=0.01$ and $y_R=10$, respectively, so we consider $x
\in [x_L, x_R]$ where $x_L=\ln y_L, x_R=\ln y_R$. 
Discretization steps were chosen as
$h=0.1$ and $k=0.1 h^2$. Concerning boundary conditions, we use the
Robin b.c. with $d=1$ on the left boundary and the Neumann b.c on
the right boundary.
\begin{figure}
\begin{center}
\includegraphics[width=0.4\textwidth]{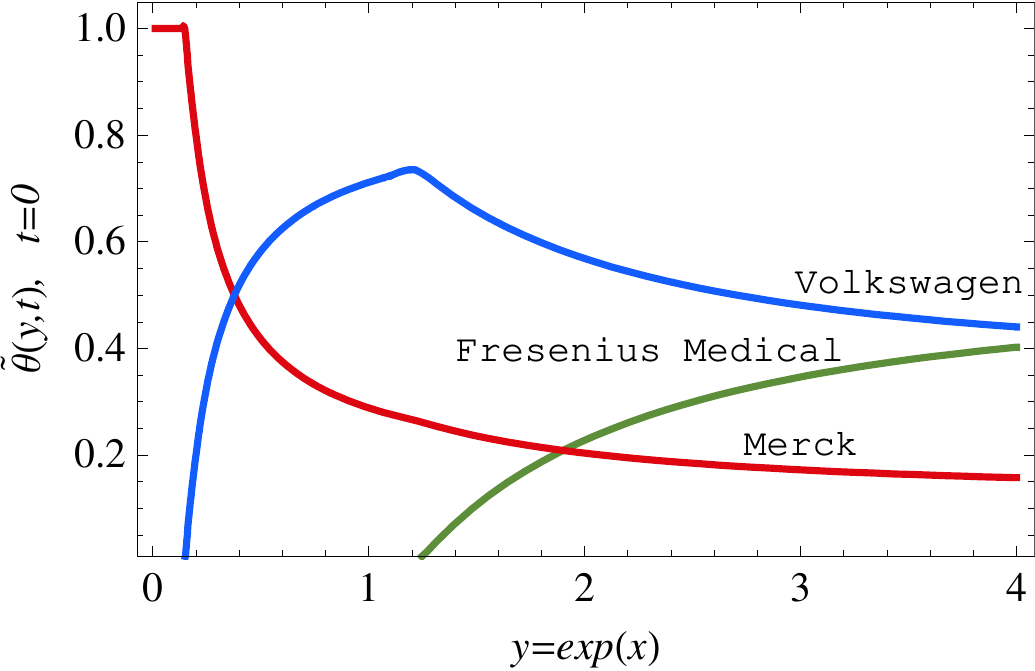}
\includegraphics[width=0.4\textwidth]{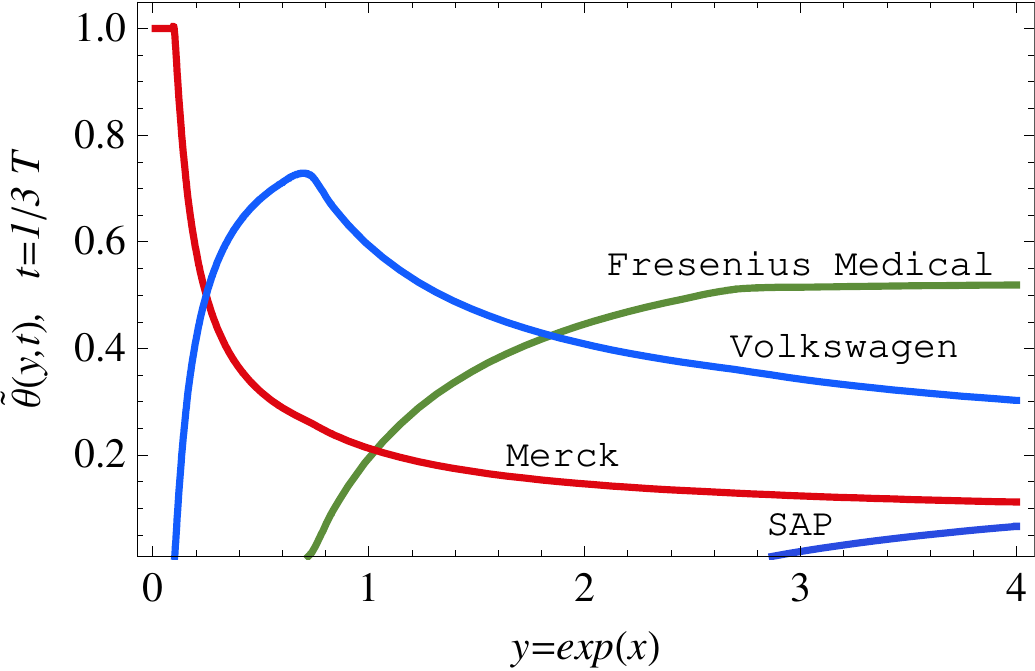}
\\
\includegraphics[width=0.4\textwidth]{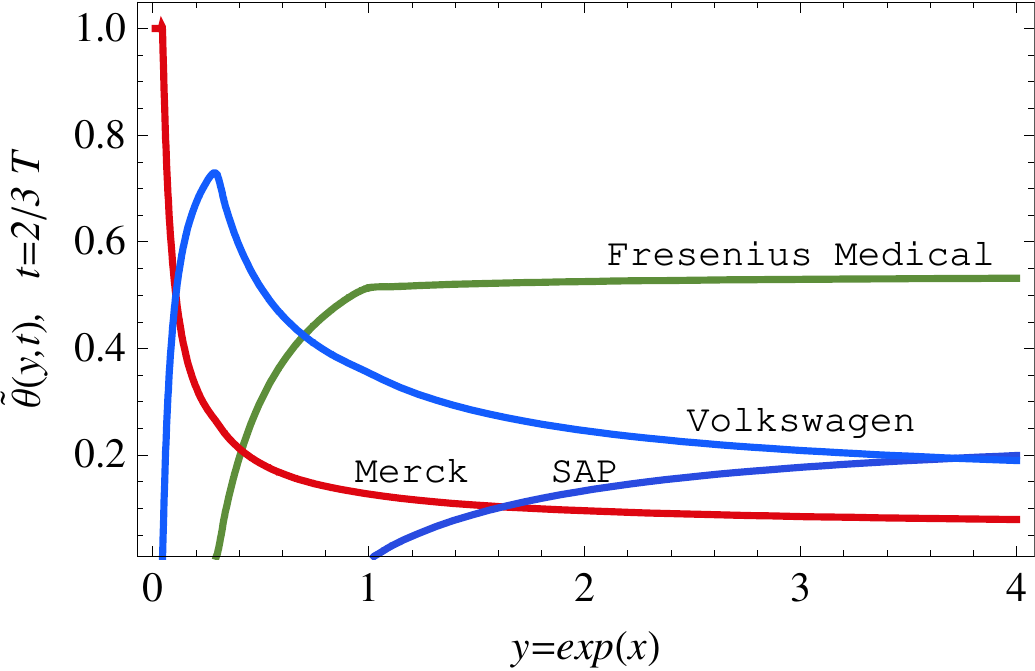}
\includegraphics[width=0.4\textwidth]{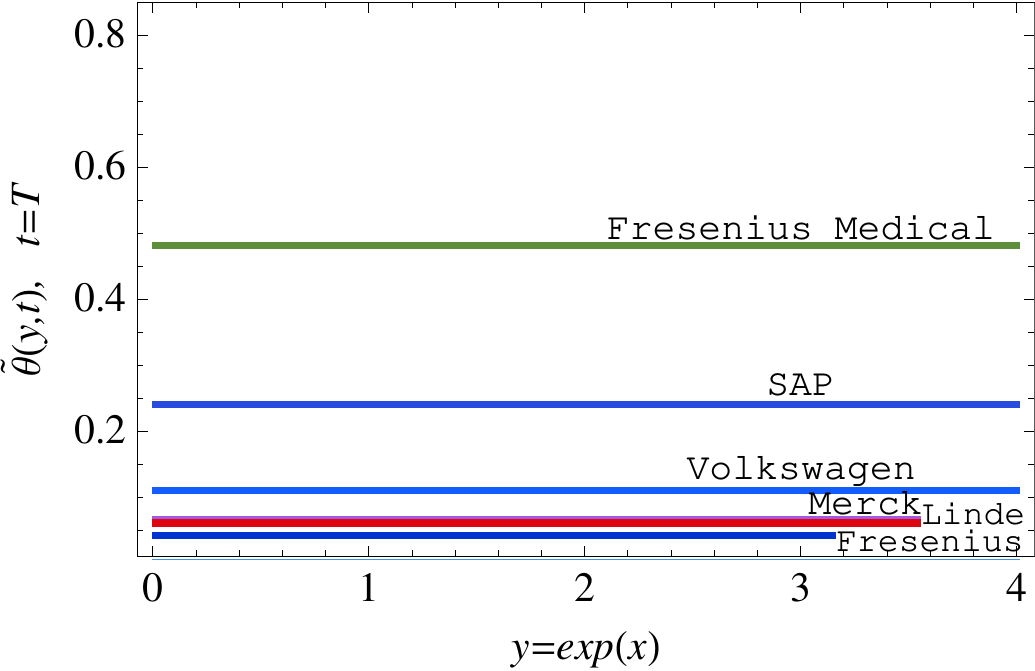}
\end{center}
\caption[caption]{Optimal response strategy $\tilde\bmtheta(y,t)$
for the DAX portfolio optimization, for time instances $t=0$,
$t=T/3$, $t= 2T/3$ and $t=T$ where $T=10$.}
\label{fig:exampleDAX}
\end{figure}

Fig.~\ref{fig:exampleDAX} shows that there are only a few relevant assets out
of the set of thirty assets entering the Index.
Tab.~\ref{tab:DAX5data} summarizes historical average returns and
covariance matrix for these assets. The figure reveals the highest
portion of Merck stocks for the early period of saving and for low
account values $y$. It is indeed reasonable to invest in an asset
with the highest expected return, although with the highest
volatility, when the account value is low, in early times of saving. Evident fast
decrement of the Merck weight can be observed for increasing
account value. Fresenius Medical has the lowest volatility out of
the considered five assets (and third lowest out of all thirty
assets) and third best mean return, which is reflected in its
major representation in the portfolio.

\begin{table}[bt]
\begin{center}
\small
\begin{tabular}{l || l | l | l | l | l |l || l }
$\bmSigma^{part}$ & Merck & VW & SAP & Fres Med & Linde & Fres &
Mean return
\\ \hline \hline
Merck & 1.6266 & -0.0155 & -0.0104 & -0.0146 & -0.0017 & -0.0033&
0.7315
\\ VW & -0.0155 & 0.1584 & 0.0345 & 0.0292 &
0.0569 & 0.0238 & 0.3413
\\ SAP & -0.0104 & 0.0345 & 0.0516
 & 0.0183 & 0.0240 & 0.0143&  0.1877
 \\ Fres Med & -0.0146 & 0.0292 &
0.0183
 & 0.0434 & 0.0227 & 0.0248 & 0.2202
\\ Linde &
-0.0017 & 0.0569 & 0.0240 & 0.0227 & 0.0530 & 0.0201 & 0.1932
\\ Fres &
-0.0033 & 0.0238 & 0.01430 & 0.0248 & 0.0201 & 0.0386 & 0.1351
\end{tabular}
\end{center}
\caption{
\small
The covariance matrix $\bmSigma^{part}$ and mean returns for six
stocks of the DAX 30 Index: Merck, Volkswagen, SAP, Fresenius
Medical, Linde, Fresenius. Based on historical data, August
2010--April 2012. Source: finance.yahoo.com} \label{tab:DAX5data}
\end{table}

In Section
\ref{ex:alpha_discontin} we showed that the sets of active indices
can be identified directly from the function
$\alpha^{\prime\prime}(\varphi)$. Moreover, based on Proposition
\ref{th:CompPsi}, there is an upper bound on investor's
coefficient of absolute risk aversion $a(x,t)$ given by $\varphi^+ -1$. When the utility function is
given as in (\ref{eq:ara2}), we have $\varphi^+ =a + 1=10$ and
so $\varphi(x,t)\le 10$ for all $x$ and $t$. Hence, only the
interval $[0,\varphi^+]$ gives relevant information for the
investor. Knowing the sets of active indices computed for
$\varphi\in[0,\varphi^+]$, the investor knows the set
$\bigcup_{\varphi\in (0,\varphi^+]} \{i\ |\  \hat{\theta}_i(\varphi)
>0\}$, i.e. the set of assets which will be entering the optimal
portfolio with a nonzero weight. To identify  the set $\{i\ |\ \hat{\theta}_i(\varphi) >0\}$ 
on a particular interval, it is
enough to calculate the optimal $\bmtheta(\varphi)$ in one single
point from the given interval.

\section*{Conclusions}
We proposed and analyzed a method of the Riccati transformation
for solving a class of Hamilton-Jacobi-Bellman equations arising
from a problem of optimal portfolio construction. We derived a
quasi-linear backward parabolic equation for the coefficient of
relative risk aversion corresponding to the value function - a
solution to the original HJB equation. Using Schauder's theory we
showed existence and uniqueness of classical H\"older smooth
solutions. We also derived useful qualitative properties of the
value function of the auxiliary parametric quadratic programming
problem after the transformation. A fully implicit iterative numerical scheme based on finite volume approximation has been proposed and numerically tested. 
We also provided a practical example of the German DAX 30
Index portfolio optimization.

\section*{Acknowledgments}
We are thankful to professor Milan Hamala  for stimulating discussions 
on parametric quadratic programming. This research was supported by the VEGA project 1/2429/12 (S.K.) and EU Grant Program FP7-PEOPLE-2012-ITN  STRIKE - Novel Methods in Computational Finance, No. 304617 (D.\v{S}.).


\end{document}